\documentclass[doublecolumn]{IEEEtran}
\usepackage{amsmath,amsthm,amsfonts,amssymb,bm,mathrsfs}
\usepackage{subfigure}
\usepackage{graphicx,multicol}
\usepackage{algorithm}
\usepackage{algpseudocode}
\usepackage{algorithmicx}
\usepackage{multirow,booktabs,threeparttable}
\usepackage{color,cite,url}
\usepackage{booktabs}
\usepackage{amssymb}
\usepackage{caption}
\usepackage{tabularx}
\usepackage{makecell}
\usepackage{enumitem}
\usepackage{dsfont}

\usepackage{wasysym}
\usepackage{listings}
\usepackage{xcolor}
\colorlet{punct}{red!60!black}
\definecolor{background}{HTML}{EEEEEE}
\definecolor{delim}{RGB}{20,105,176}
\colorlet{numb}{magenta!60!black}

\lstdefinelanguage{json}{
  basicstyle=\ttfamily\small,
  numbers=left,
  numberstyle=\tiny,
  stepnumber=1,
  numbersep=5pt,
  showstringspaces=false,
  breaklines=true,
  backgroundcolor=\color{background},
  literate=
     *{0}{{{\color{numb}0}}}{1}
      {1}{{{\color{numb}1}}}{1}
      {2}{{{\color{numb}2}}}{1}
      {3}{{{\color{numb}3}}}{1}
      {4}{{{\color{numb}4}}}{1}
      {5}{{{\color{numb}5}}}{1}
      {6}{{{\color{numb}6}}}{1}
      {7}{{{\color{numb}7}}}{1}
      {8}{{{\color{numb}8}}}{1}
      {9}{{{\color{numb}9}}}{1}
      {:}{{{\color{punct}{:}}}}{1}
      {,}{{{\color{punct}{,}}}}{1}
      {\{}{{{\color{delim}{\{}}}}{1}
      {\}}{{{\color{delim}{\}}}}}{1}
      {[}{{{\color{delim}{[}}}}{1}
      {]}{{{\color{delim}{]}}}}{1},
}

\newtheorem{theorem}{Theorem}
\newtheorem{lemma}{Lemma}
\newtheorem{corollary}{Corollary}
\newtheorem{assumption}{Assumption}
\newtheorem{remark}{Remark}

\newcolumntype{Y}{>{\centering\arraybackslash}X}
\newcolumntype{L}{>{\raggedright\arraybackslash}X}

\graphicspath{{fig/}}

\newtheoremstyle{break}
  {\topsep}{\topsep}
  {\itshape}{}%
  {\bfseries}{.}{\newline}
  {}%

\usepackage{setspace}
\theoremstyle{remark}

\theoremstyle{plain}

\begin{document}
\title{Optimizing NetGPT via Routing-Based Synergy and Reinforcement Learning}

 \author{
 	Yuxuan Chen, Rongpeng Li, Xianfu Chen, Celimuge Wu, Chenghui Peng, Zhifeng Zhao, and Honggang Zhang

 	\thanks{Y. Chen and R. Li are with Zhejiang University, Hangzhou 310027, China, (email: \{cyx00, lirongpeng\}@zju.edu.cn).}

    \thanks{X. Chen is with the Shenzhen CyberAray Network Technology Co., Ltd, China (e-mail: xianfu.chen@ieee.org).}

    \thanks{C. Wu is with The University of Electro-Communications, Japan (e-mail: celimuge@uec.ac.jp).}

    \thanks{C. Peng is with Huawei Technologies Co., Ltd., Shanghai 210026, China (email: pengchenghui@huawei.com).}

    \thanks{Z. Zhao is with Zhejiang Lab, Hangzhou 310012, China as well as Zhejiang University, Hangzhou 310027, China (email: zhaozf@zhejianglab.org).}
        
    \thanks{H. Zhang is with Macau University of Science and Technology, Macau 999078, China (email: hgzhang@must.edu.mo).}
        
}

\maketitle

\begin{abstract}
Large language model (LLM) agents at the network edge offer low-latency execution for routine queries. In contrast, complex requests often require the superior capability of cloud models, incurring higher latency and cost. To navigate this quality-cost trade-off under dynamic network conditions, we propose a cloud-edge synergy for NetGPT that integrates network-aware routing with on-edge self-improvement. Specifically, our framework routes structured tool-calling requests to cloud or edge agents via a novel scoring policy. We prove that, under mild regularity assumptions, the optimal routing rule admits a unique fallback threshold with monotone dependence on bandwidth and round-trip time (RTT). Concurrently, based on the dataset collected from requests routed to the cloud and corresponding responses, we instantiate a schema-preserving reinforcement learning (RL) to improve the capability of the edge agent. We analyze a supervised finetuning (SFT)-anchored composite objective that combines a reverse-KL trust-region step with a forward-KL realignment toward the SFT prior, explaining stability and constraining policy drift. 
Both the network-aware routing policy and the edge agent are updated coherently. 
Experiments across controlled network states and pricing schedules demonstrate smooth quality-cost frontiers, consistent gains of dynamic fallback thresholds over fixed policies, and sustained reductions in offloading while maintaining task success and schema-correct outputs.
\end{abstract}


\section{Introduction}
The heterogeneous deployment of Large Language Models (LLMs) across the network is forming a NetGPT \cite{Chen2024NetGPT}. In this architecture, lightweight LLMs at the edge offer low latency but may struggle with complex queries, while more powerful and general-purpose cloud models lead to significant inference inefficiency and economic expenditure. The integration of external tools introduces even greater time-sensitivity \cite{OpenAI2023FunctionCalling} than vanilla LLM interactions \cite{Yao2023ReAct}. Consequently, there is a compelling need for a cloud-edge synergy for NetGPT. Such a collaboration promises low-latency interactions at the edge with on-demand access to the high-quality reasoning of the cloud \cite{Shi2016Edge, Li2018Edgent}. 

While devising a simple binary routing scheme (i.e., edge or cloud) sounds straightforward, an effective cloud-edge synergy is far more complicated. We face three key challenges. First, routing decisions are inherently state-dependent: the incremental quality from escalating a \emph{hard} query to a stronger cloud model must be weighed against \emph{network conditions} and \emph{inference latency} \cite{Mach2017MECsurvey}. Second, preference data collected from interacting with the cloud LLM can naturally be used to supervise the edge LLM \cite{Ouyang2022InstructGPT} for capability improvement. Therefore, the developed routing policy should have the capability of automatically adapting to dynamic network conditions and concept drift \cite{Besbes2014NonStationary}. Third, LLM agents must adhere to the tool-calling format (JSON) \cite{Schick2023Toolformer}. However, continuous improvement risks eroding schema adherence \cite{Agarwal2025ThinkInsideJSON}.

Existing paradigms cannot simultaneously solve these challenges. For example, cost-aware cascades \cite{chen-etal-2024-frugalgpt,lu-etal-2024-routing, Zheng2025DiSRouter} and preference-trained routers \cite{ong-etal-2025-routellm, ding-etal-2024-hybridllm} reduce average cost by sending \textit{easy} queries to smaller models and escalating \textit{hard} ones, but they typically operate as front-door or stage-wise mechanisms that are agnostic to link dynamics and seldom support continuous online improvement under schema constraints.
Representation-learning routers further move beyond shallow classifiers: contrastive designs such as RouterDC \cite{chen2024routerdc} learn query-conditioned embeddings to score per-model utility, while graph-based approaches like GraphRouter \cite{Feng2025GraphRouter} cast selection as inductive edge prediction on query-task-model heterogeneous graphs, improving generalization to unseen models and tasks, yet still treat latency and schema validity as exogenous variables.
Distributed or multi-agent variants \cite{yue-etal-2025-masrouter} emphasize ensemble specialization and coordination yet still decouple decisions from explicit network signals or rely on fixed acceptance rules. 
On the other hand, reinforcement learning (RL)-based online alignment pipelines — such as online direct preference optimization (DPO) \cite{Qi2024OnlineDPO} and iterative RL with human feedback (RLHF) workflows \cite{dong2024rlhf_workflow} — demonstrate how to continually refresh reward models and policies from streaming feedback under stability anchors, but they are typically designed for monolithic LLMs rather than tool-calling cloud-edge agents.
Consequently, there still exists a methodological gap in optimizing NetGPT.

\begin{table*}[t]
\small
\centering
\setlength{\tabcolsep}{3pt}      
\renewcommand{\arraystretch}{1.12}
\caption{\textsc{The Summary of Differences with Related Literature.}}
\label{tab:related_work}
\begin{tabular}{
  p{0.08\textwidth}  
  p{0.08\textwidth} 
  p{0.08\textwidth}  
  p{0.08\textwidth} 
  p{0.08\textwidth}  
  p{0.08\textwidth} 
  p{0.34\textwidth} 
}
\hline
\textbf{References} &
\centering\textbf{Multi-} \\ \centering\textbf{stage} &
\centering\textbf{Tool} \\ \centering\textbf{Schema-} \\ \centering\textbf{aware} &
\centering\textbf{Pref.} \\ \centering\textbf{Quality} &
\centering\textbf{Network} \\ \centering\textbf{Aware} &
\centering\textbf{Online} \\ \centering\textbf{improv.} &
\textbf{Brief description} \\
\hline
\cite{ong-etal-2025-routellm, ding-etal-2024-hybridllm, chen2024routerdc, Feng2025GraphRouter} &
\centering \Circle & \centering \Circle & \centering \CIRCLE & \centering \Circle & \centering \Circle &
Single-shot, front-door routing: a preference-trained scorer picks the model before generation. \\
\hline
\cite{chen-etal-2024-frugalgpt, lu-etal-2024-routing, wang2025mixllm} &
\centering \CIRCLE & \centering \Circle & \centering \CIRCLE & \centering \Circle & \centering \Circle &
Multi-stage cascade: apply a fixed quality threshold after each LLM to early-exit or escalate. \\
\hline
\cite{metallm} &
\centering \Circle & \centering \Circle & \centering \CIRCLE & \centering \Circle & \centering \CIRCLE &
Contextual-bandit router: adapts model selection online under cost/quality uncertainty; agnostic to link dynamics and schema. \\
\hline
\cite{Zheng2025DiSRouter} &
\centering \CIRCLE & \centering \Circle & \centering \CIRCLE & \centering \Circle & \centering \CIRCLE &
No explicit router: each LLM decides to forward/reject; trained for system utility. \\
\hline
\cite{yue-etal-2025-masrouter} &
\centering \Circle & \centering \CIRCLE & \centering \CIRCLE & \centering \Circle & \centering \Circle &
Multi-agent routing with schema-aware coordination across specialized agents. \\
\hline
\textbf{This work} &
\centering \textbf{\CIRCLE} & \centering \textbf{\CIRCLE} & \centering \textbf{\CIRCLE} & \centering \textbf{\CIRCLE} & \centering \textbf{\CIRCLE} &
Cloud-edge, tool-schema-aware multi-stage routing with state-dependent fallback thresholds enabling online adaptation. \\
\hline
\multicolumn{7}{r}{\footnotesize \textit{Notation:} \Circle\ indicates not included; \CIRCLE\ indicates fully included.} \\
\end{tabular}
\end{table*}
We address this gap by proposing a cloud-edge LLM pipeline that unifies tool-calling, network-aware routing, and online adaptation. Concretely, we use a unified router score trained from preference and quality signals. The score is produced by a lightweight reward model (RM) and is periodically refreshed using cached data from interactions with the cloud. Based on this score, a state-dependent fallback threshold is derived from network states like measured round-trip time (RTT) and bandwidth. Simultaneously, the edge LLM policy is also updated via RL \cite{Schulman2017PPO} with a cross-entropy supervised finetuning (SFT) anchor that preserves schema-correct tool-calling while ensuring stable updates. Furthermore, we analyze the theoretical simplicities and empirical advantages of the proposed pipeline. 
While highlighting the key differences with existing works in Table \ref{tab:related_work}, the main contributions of this paper are summarized as follows.
\begin{itemize}
    \item We cast cloud-edge routing as a unified score-threshold policy whose acceptance boundary is an explicit function of measured network RTT and bandwidth. Under mild regularity assumptions on the confidence score, we theoretically prove that the optimal offloading rule admits a unique state-dependent fallback threshold and can be characterized by a first-order balance between the cloud's marginal quality gain and its marginal cost.
    This yields monotone comparative statics to derive a network-aware fallback threshold.
    \item We unify tool-schema adherence with continual improvement by (i) using cached queries as on-policy supervision to refresh a lightweight RM, and (ii) optimizing the edge LLM policy via RL with an SFT anchor to maintain schema-correct tool-calling and stability during updates. By keeping routing decisions aligned with the evolving quality-latency-cost landscape, our design effectively tackles the underlying structural drift. 
    \item We present an end-to-end evaluation under controlled network and pricing regimes, showing that dynamic, network-aware fallback thresholds consistently dominate fixed policies on the quality-cost frontier. Meanwhile, SFT-anchored on-device RL contributes to preserving task success and schema-correct tool-calling, and periodic RM refresh keeps the router calibrated and improves routing accuracy. 
\end{itemize}

\begin{table}[t]
\centering
\caption{Notations used in the paper.}
\label{tab:Notations}
\footnotesize
\renewcommand{\arraystretch}{1.1}
\setlength{\tabcolsep}{4pt}
\begin{tabular}{>{\centering\arraybackslash}p{2.2cm} p{6cm}}
\toprule[0.75pt]
\textbf{Notation} & \textbf{Description} \\
\midrule[0.5pt]

$x_i$ & Task/request index $i$ \\
$H_{i,k}$ & Context of task $i$ at step $k$ (history, tool outputs, etc.) \\
$u_{i,k},\,u^{\star}_{i,k}$ & Structured actions produced respectively by the edge LLM and the cloud model (\texttt{tool}, \texttt{args}, \texttt{thought}) \\
$g_\psi$ & Reward model (router score model) with parameters $\psi$ \\
$s_{i,k}$ & Router/reward-model score at $(i,k)$ \\
$\tau(\cdot)$ & Network-aware fallback threshold (runtime: $\tau(\widehat S_{i,k})$) \\
$d_{i,k}\!\in\!\{E, C\}$ & Routing decision at $(i,k)$: edge accept (E) or cloud offload (C) \\
$S_{i,k},\widehat{S}{i,k}$ & Latent network state at step $k$ and its one-step-ahead estimate used in $\tau(\widehat S{i,k})$ \\
$\mathrm{RTT}_{i,k}$ & Baseline round-trip path latency (propagation/queuing; size-independent) at step $k$ \\
$\mathrm{BW}_{i,k}$  & Effective bandwidth at step $k$ (determines transmission time $T_{\mathrm{tx}}$) \\
$\varepsilon_{i,k},\,\Sigma$ & Gaussian jitter and its scale in the Gauss–Markov model for time-varying network state \\
$L_E,\,L_C$ & Latency of a local step (edge) and a cloud-offloaded step (uplink-compute-downlink) \\
$L_{i,k},\,L(x_i)$ & Realized latency at $(i,k)$; end-to-end latency of task $i$ \\
$N_{i,k}$ & Token usage of the cloud model at $(i,k)$ \\
$C_{i,k},\,C(x_i)$ & Step cost and total task cost (latency penalty + cloud cost) \\
$Q_{i,k},\,Q(x_i)$ & Step quality in $[0,1]$ and average task quality \\
$J(x)$ & Utility measuring the quality-cost trade-off \\
$\lambda$ & Trade-off coefficient on cost \\
$\pi_\theta$ & Edge LLM policy with parameters $\theta$ \\
$\pi_{\text{SFT}}$ & Supervised reference policy used for SFT anchoring \\
$r_{i,k}$ & Learning reward for online updates at $(i,k)$ \\
$\kappa(S)$ & Network scaling of cloud cost under state $S$ in comparative statics ($\Delta C_S=\kappa(S)\Delta C$) \\
$\rho(s)$ & Local benefit–cost ratio at score $s$ \\
$f(s)$ & Density of the confidence score $s$ \\
$C_{\text{tok}}$ & Unit price per cloud token used in the monetary-cost term \\
$\alpha_{\mathrm{RTT}},\,\beta_{\mathrm{BW}},\,\gamma_{\mathrm{hist}}$ & Coefficients in the linear dynamic fallback threshold (RTT/BW/history sensitivities) \\
$f_\omega(\cdot),\,\omega$ & Lightweight neural router (PolicyNet) and its parameters \\
$\gamma$ & Discount factor in the on-device RL objective \\
$\rho_{i,k}$ & PPO importance ratio \\
$A_{i,k}$ & Advantage estimate used in PPO\\
$\epsilon$ & PPO clipping range in the surrogate objective \\
$\beta_{\mathrm{KL}}$ & Reverse-KL penalty coefficient in PPO \\
$\eta$ & Forward-KL anchor (SFT realignment) weight \\
$\pi_t$ & Current policy at iteration $t$ in the two-stage update \\
$\mathcal{T}_{i,k}$ & Cached tuple $(x_i, H_{i,k}, u_{i,k}, u^{\star}_{i,k}, s_{i,k})$\\
$\mathcal{B}_{\mathrm{RM}},\,\mathcal{B}_{\mathrm{RL}}$ & Caches storing on-policy samples for RM and PPO updates \\

\bottomrule[0.75pt]
\end{tabular}
\end{table}

The remainder of the paper has been organized as follows. Section~\ref{sec:related} reviews the background and related work. Section~\ref{sec:Problem} describes the system model and formulates the optimization problem. Section~\ref{sec:method} introduces the proposed dynamic fallback threshold mechanism and online self-improvement framework from both operational and theoretical perspectives. 
Section~\ref{sec:evaluation} presents the experimental results. Section~\ref{sec:conclusion} concludes the paper.

Beforehand, major notations used throughout the paper are summarized in Table \ref{tab:Notations}.

\section{Related Work}
\label{sec:related}

\subsection{LLM Routing and Selective Inference}
Model selection across heterogeneous LLMs has rapidly evolved from early cost-aware cascades and difficulty-aware routers to theoretically grounded and benchmarked systems. Prototype cascades (e.g., \cite{chen-etal-2024-frugalgpt}) reduce cost by answering \textit{easy} queries with cheaper models and escalating \textit{hard} cases to stronger models, showing sizable cost-quality gains but relying on offline heuristics and static thresholds. \cite{ding-etal-2024-hybridllm} formalizes two-model routing via a learned difficulty predictor that trades quality for cost and reports up to 40\% fewer calls to the large model with negligible quality loss. On the supervised front, \cite{ong-etal-2025-routellm} trains a Bradley-Terry-Luce-style router on preference pairs mined from multi-model comparisons and evaluates routing under a cost/willingness-to-pay knob, and \cite{Hu2024RouterBench} standardizes multi-LLM routing evaluation and highlights the centrality of calibrated quality estimates. 

Representation-learning routers broaden beyond shallow classifiers: \cite{chen2024routerdc} uses dual-contrastive training to predict per-model utility. \cite{Feng2025GraphRouter} casts selection as inductive edge prediction on a query-model-task heterogeneous graph to improve generalization to unseen models and tasks. Concurrently, cascading receives theoretical treatment: \cite{dekoninck2025a} presents a unified analysis proving optimal strategies for both routing and cascades and introduces cascade routing, which outperforms either paradigm given accurate quality (and cost) estimators. Empirically oriented works refine the design space—e.g., \cite{yue2024large} explores MIX-based cascades for reasoning workloads, \cite{shen-etal-2025-sater} proposes dual-mode routing that blends pre-generation and cascade policies, and \cite{ding2025bestroute} augments routing with adaptive additional computation for small models to yield test-time optimal compute. Bandit and non-parametric perspectives further question router complexity: \cite{metallm} frames selection as a contextual bandit to adapt to cost/quality uncertainty, while \cite{li2025rethinking} shows strong performance of kNN-based routing versus learned routers on standardized benchmarks, underscoring the role of locality in embedding space.
\cite{yue-etal-2025-masrouter} extends routing to multi-agent settings by jointly choosing collaboration modes, allocating roles, and routing queries across heterogeneous LLMs, but still abstracts away time-varying network effects and offers limited online adaptation under bandit feedback—constraints that reduce suitability for cloud-edge deployments.

Despite progress, several structural limitations persist for our cloud-edge setting. First, network unawareness is pervasive: most studies treat \textit{cost} as API price or FLOPs and either ignore latency/bandwidth or treat them as constants; even works that model cost explicitly rarely integrate time-varying link measurements into the decision boundary, and recent surveys acknowledge that many routers ignore latency in practice \cite{wang2025mixllm}. Second, online self-improvement is rarely a primary design goal: routers are commonly trained offline with full supervision (quality labels for all models) and lack the capability of online adaptation. Third, schema-constrained tool-calling—now standard in agentic pipelines—poses a stability constraint largely absent from routing benchmarks. Motivated by these gaps, we propose a network-aware, schema-preserving routing pipeline with continual improvement. It uses a unified score to govern edge acceptance and offload, adapts fallback thresholds explicitly to measured RTT and bandwidth, and reuses cached escalated queries to refresh the router's reward model and the edge LLM policy.

\subsection{Online Learning and Adaptive Optimization in Networked Agents}
In the LLM era, online alignment seeks to improve models during deployment under bandit feedback, distribution shift, and tight latency, rather than relying solely on offline preference corpora. \cite{calandriello2024human} further formalizes online preference optimization with equivalences to Nash mirror descent and demonstrates how to update policies from streaming preferences with Kullback-Leibler (KL) control. To curb catastrophic forgetting in continual DPO, \cite{Qi2024OnlineDPO} introduces an online DPO method with a fast-slow memory that decouples rapid adaptation from a stable anchor to preserve earlier competencies.
Beyond pure DPO, self-play and self-rewarding methods close the loop by generating rewards on the fly. \cite{Yuan2024SelfRewarding} demonstrates that LLM-as-a-Judge can bootstrap iterative DPO and jointly improve both the model and its internal judge. To mitigate judge drift, \cite{wang2025cream} adds consistency regularization across iterations, improving reward reliability in self-rewarding pipelines. A game-theoretic perspective is provided by \cite{wu2025selfplay}, which frames alignment as a constant-sum game and proves convergence of iterative self-play preference optimization updates.

Several works study robustness and grouping in continual alignment. \cite{ramesh2024group} seeks worst-group improvements to avoid regressing minority preferences when updating online. On the RL side, \cite{shao2024deepseekmath} presents a group-relative PPO variant that removes the critic and estimates baselines from group scores, offering a lighter-weight alternative for iterative preference-driven improvement in reasoning-heavy tasks.
From a systematic perspective, \cite{dong2024rlhf_workflow} describes online iterative RLHF workflows that construct proxy preference models, generate on-policy data, and repeat align-evaluate-deploy cycles, reporting consistent gains over offline pipelines. A complementary theoretical analysis is provided by \cite{Ye2024OnlineIterativeRLHF}, which generalizes online RLHF to preference oracles beyond Bradley-Terry assumptions and develops sample-efficient procedures for querying and updating on the fly.

Despite these advances, a significant gap remains for tool-oriented, cloud-edge agents. First, the coupling between routing and learning implies that the same scalar signal acts as an accept-reject score at the edge and a reward for online adaptation. Nevertheless, the aforementioned decoupled design could compromise closed-loop improvement at deployment. Second, schema-preserving constraints, now standard via JSON-schema structured outputs, are typically absent from online preference updates even though they are essential for maintaining the tool-calling validity under latency budgets. Our approach addresses both by using a router-as-reward score to drive state-dependent acceptance and on-device updates, and by anchoring updates with an SFT anchor to ensure schema-correct tool-calling and stable improvement. 

\section{System Model and Problem Formulation}
\label{sec:Problem}

\begin{figure}
	\centering
	\includegraphics[width=0.495\textwidth]{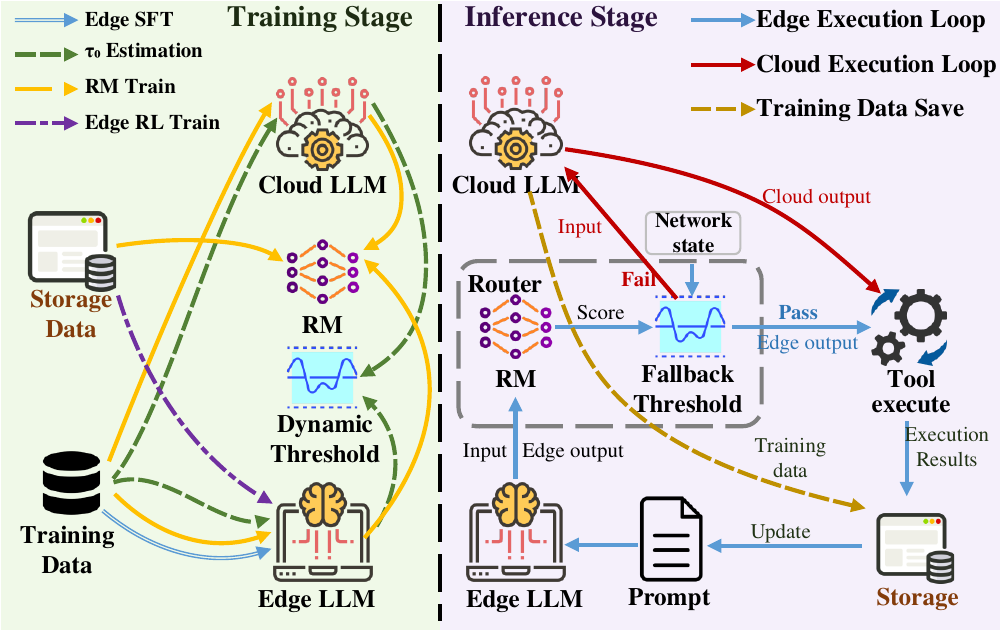}
	\caption{Overview of the proposed cloud-edge pipeline.}
	\label{fig:system_simple}
\end{figure}

\subsection{System Model}
\label{subsec:system_model}
Fig. \ref{fig:system_simple} provides a high-level overview of the cloud-edge synergy pipeline in NetGPT, where in response to an input query $x_i$, the model alternates between internal reasoning and tool invocation at each step.

\begin{figure*}
	\centering	\includegraphics[width=0.98\textwidth]{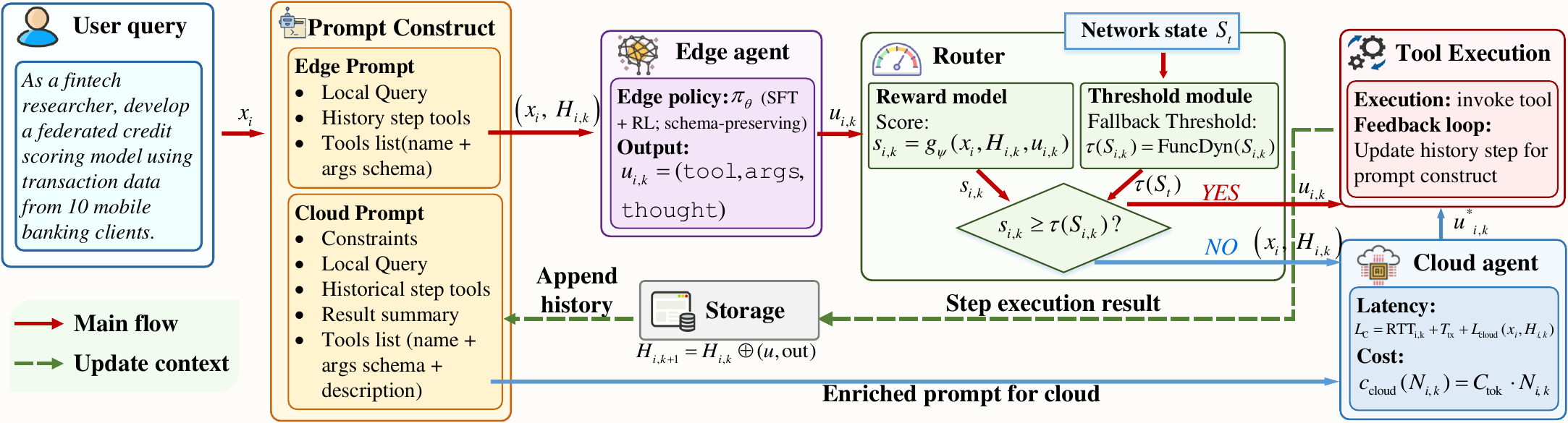}
	\caption{Online-stage of the proposed pipeline.}
	\label{fig:online_framework}
\end{figure*}
As shown in Fig.~\ref{fig:online_framework}, at each step $k$, the current context $H_{i,k}$ (comprising the accumulated tool outputs and reasoning traces up to step $k$) is combined with the query $x_i$ to form the edge prompt. The edge LLM policy $\pi_\theta$ then samples an action:
\begin{align}
u_{i,k} =(\texttt{tool},\texttt{args},\texttt{thought}) \sim \pi_\theta(\cdot \mid x_i, H_{i,k}). 
\end{align}
Here $u_{i,k}$ represents both the model's reasoning and its proposed tool operation.
Correspondingly, the reward model (RM) $g_\psi$ evaluates $u_{i,k}$ and produces a scalar score $s_{i,k}$:
\begin{align}
  s_{i,k}=g_\psi(x_i, H_{i,k}, u_{i,k})\in\mathbb{R}.
  \label{eq:reward}
\end{align}
Subsequently, the routing decision can be made by comparing the score with a network-aware fallback threshold $\tau(S_{i,k})$ with $S_{i,k}$ denoting the networking state. In other words, if $s_{i,k} \geq \tau(S_{i,k})$, the prompt will be solely handled at the edge; otherwise, the more powerful cloud LLM produces a refined output:
\begin{align}
u^{\star}_{i,k} \sim \pi_{\text{cloud}}(\cdot \mid x_i, H_{i,k}).
\end{align}
Lately, $u^{\star}_{i,k}$ serves as the preferred answer for both routing supervision and reward calibration.
In our work, we assume the edge LLM policy $\pi_\theta$, the RM $g_\psi$ and the fallback threshold $\tau(S_{i,k})$ has been initialized in Appendix \ref{sec:initialization}, while the cached tuple $\mathcal{T}_{i,k}\!\equiv\!(x_i, H_{i,k}, u_{i,k}, u^{\star}_{i,k}, s_{i,k})$ will be appended to the RM cache $\mathcal{B}_{\mathrm{RM}}$ and the RL cache $\mathcal{B}_{\mathrm{RL}}$ for online self-improvement, as detailed in Appendix \ref{sec:log}.

\subsubsection{Network State Model}
Following ITU-T Y.1541 delay/QoS objectives for IP networks \cite{ITU-T-Y1541-2011} and IETF RFC 6349's bandwidth--delay-product rationale for TCP throughput measurement \cite{rfc6349}, we model the \textit{network state} at step $k$ of query $i$ with two primary observables:
\begin{align}
S_{i,k} = \big(\mathrm{RTT}_{i,k},\ \mathrm{BW}_{i,k}\big),
\end{align}
where $\mathrm{RTT}_{i,k}$ denotes round-trip time and $\mathrm{BW}_{i,k}$ is bandwidth.

\subsubsection{Inference Cost and Quality Model}

As mentioned earlier, we formalize the routing decision for step $(i,k)$ as:
\begin{align}
d_{i,k} =
\begin{cases}
\text{Edge}, & s_{i,k}\ge \tau(\cdot), \\[4pt]
\text{Cloud}, & s_{i,k}< \tau(\cdot).
\end{cases}
\label{eq:route}
\end{align}

\noindent In words, $d_{i,k}$ selects \textsc{Edge} when the router score $s_{i,k}$ exceeds the fallback threshold $\tau(\cdot)$, and \textsc{Cloud} otherwise. Given this decision $d_{i,k}$, we model the step latency as:
\begin{align}
L_{i,k}=
\begin{cases}
L_{\mathrm{E}}(x_i,H_{i,k}), & d_{i,k}=\text{Edge}, \\[2pt]
L_{\mathrm{C}}(x_i,H_{i,k},S_{i,k}), & d_{i,k}=\text{Cloud}.
\end{cases}
\label{eq:latency}
\end{align}
Here, $L_{\mathrm{E}}$ is the on-device inference latency; $L_{\mathrm{C}}$ is the end-to-end offloading latency under the network state $S_{i,k}$.  
We decompose the cloud case into propagation, transfer, and cloud compute:
\begin{align}
&L_{\mathrm{C}}(x_i,H_{i,k},S_{i,k}) \label{eq:cloud_latency}
\\
=& \mathrm{RTT}_{i,k}
+ T_{\mathrm{tx}}(x_i,H_{i,k},u^{\ast}_{i,k},S_{i,k}) + L_{\mathrm{cloud}}(x_i,H_{i,k}), \notag
\end{align}
where $T_{\mathrm{tx}}(\cdot)$ is the payload-dependent transmission time under limited bandwidth and $L_{\mathrm{cloud}}(\cdot)$ is the cloud LLM's computation time. 

We couple latency and monetary terms into a step-level cost:
\begin{align}
  C_{i,k} &= \alpha L_{i,k} + C_{\mathrm{tok}} \cdot N_{i,k}.
\label{eq:cloud_cost}
\end{align}
Here $\alpha\ge0$ converts latency to a penalty so that quality-cost trade-offs can be scalarized; $C_{\mathrm{tok}}$ accounts for the price per cloud token used; and $N_{i,k}$ is the number of tokens on the cloud path (and $N_{i,k}=0$ when $d_{i,k}=\text{Edge}$).

For each step $(i,k)$, we define a quality score $Q_{i,k}\in[0,1]$ that reflects how well the model's output completes the required tool invocation or reasoning. The score is computed on-device using a frozen evaluation model trained on a disjoint dataset similar to the SFT or RM data. Each sample is cross-checked by human auditing before freezing the evaluator parameters.
We formalize it as:
\begin{align}
Q_{i,k}
&=\mathbb{I}\!\left[\texttt{schema}(u_{i,k})=1\right]\;
\tilde q_\phi(x_i,H_{i,k},u_{i,k}). 
\label{eq:quality}
\end{align}
Here $\tilde q_\phi(\cdot)\in[0,1]$ is the normalized task-quality score from the evaluator, and the indicator $\mathbb{I}[\cdot]$ ensures that any structural violation ($\texttt{schema}=0$) directly forces $Q_{i,k}=0$.

\subsection{Problem Formulation}
\label{subsec:formulation}
For each task $x_i$, the task-level quality and cost are defined by aggregating per-step quantities separately:
\begin{align}
Q(x_i) &= \frac{1}{K_i}\sum_{k=1}^{K_i} Q_{i,k}, \qquad
C(x_i) = \sum_{k=1}^{K_i} C_{i,k}. 
\end{align}
We then define a utility function that balances these two metrics:
\begin{align}
J(x_i) = Q(x_i) - \lambda C(x_i), \qquad \lambda > 0. 
\end{align}
The trade-off coefficient $\lambda$ controls the tolerance for latency and cost under given quality priorities.

The overall optimization objective is to maximize the expected utility across all tasks by jointly learning the network-aware fallback threshold function $\tau(\cdot)$ and the edge LLM policy parameters $\theta$:
\begin{align}
\max_{\tau(\cdot),\,\theta}\; \mathbb{E}_{x\sim\mathcal{D}}\!\left[J(x)\right]
= \max_{\tau(\cdot),\,\theta}\; \mathbb{E}\!\left[Q(x) - \lambda C(x)\right]. 
\label{eq:target}
\end{align}
Here, $x\sim\mathcal{D}$ denotes sampling a task from the deployment-time workload distribution; $\tau(\cdot)$ governs routing under the observed network state $S$; and $\theta$ parameterizes the on-device policy that generates structured tool calls.

The key difficulty to solve \eqref{eq:target} lies in coupling short-term network adaptation with long-term model improvement. Existing formulations typically isolate routing from model adaptation \cite{ong-etal-2025-routellm, chen2024routerdc}, and often assume fixed or exogenous communication costs or static model behavior \cite{chen-etal-2024-frugalgpt, ding-etal-2024-hybridllm}. Nevertheless, decoupling of network-aware fallback thresholds and underlying model improvement could significantly compromise the efficiency. Instead, the fallback threshold $\tau(\cdot)$ should react to instantaneous variations of the network state $S$ while the edge LLM $\pi_\theta$ and RM $g_\psi$ continuously improve according to observed data during online inference. These two processes operate on different time scales but jointly determine the system's stability and efficiency. 

\section{Method}
\label{sec:method}
To maximize the expected utility $J$, we combine dynamic fallback threshold adaptation with schema-preserving RL, both guided by a unified router score. The former responds rapidly to changing network states, while the latter gradually improves the edge LLM's competence under data drift without compromising tool-schema correctness.

\subsection{Dynamic Fallback Thresholds}
\label{subsec:dynamic_th}
To characterize the routing behavior under varying network states, we first analyze the threshold-based decision rule and its optimality conditions. Without loss of generality, for each request with confidence score $s$, we assume that the edge execution yields quality $Q_E$ and cost $C_E$, while the cloud execution yields $Q_C$ and $C_C$.
We define the marginal cloud-edge differences: 
\begin{align}
    \Delta Q(s) \triangleq \mathbb{E}[\, Q_C - Q_E \mid s \,], \quad
    \Delta C(s) \triangleq \mathbb{E}[\, C_C - C_E \mid s \,].
    \label{eq:differentials}
\end{align}
Intuitively, $\Delta Q(s)$ is the expected quality gain (if any) from escalating a score-$s$ request to the cloud rather than accepting the edge inference, while
$\Delta C(s)$ is the induced extra cost (e.g., latency, transmission, cloud inference). Let $f(s)$ be the density of the confidence score $s$, and we make the following assumption.

\begin{assumption}[Regularity of the score]
\label{ass:regularity} 
The confidence score $s$ admits an absolutely continuous distribution with density $f$ that is strictly positive on the relevant support.
\end{assumption}
Assumption \ref{ass:regularity} only requires conditional continuity with respect to the score, but does not specify any particular parametric model for $(Q_E, Q_C,C_E, C_C)$. Therefore, it can be easily met.

We now express the sensitivities of $\big(Q(\tau), C(\tau)\big)$ with respect to the fallback threshold $\tau$.

\begin{lemma}[Frontier sensitivities]
\label{lem:sensitivity}
Under Assumption~\ref{ass:regularity},
\begin{align}
    \frac{dQ}{d\tau}(\tau) = f(\tau)\, \Delta Q(\tau), \qquad
    \frac{dC}{d\tau}(\tau) = f(\tau)\, \Delta C(\tau).
    \label{eq:sensitivity}
\end{align}
\end{lemma}

\begin{proof}
By the definition of the score-threshold policy,
\begin{align}
    Q(\tau) = \int_{\tau}^{\infty} \mathbb{E}[Q_E\mid s] f(s)\,ds 
    + \int_{-\infty}^{\tau} \mathbb{E}[Q_C\mid s] f(s)\,ds.
\end{align}
Applying Leibniz's rule and continuity yields
\begin{align}
    \frac{dQ}{d\tau}(\tau)
    &= -\mathbb{E}[Q_E\mid s=\tau]f(\tau) 
    + \mathbb{E}[Q_C\mid s=\tau]f(\tau)\notag \\
    &= f(\tau)\Delta Q(\tau).
\end{align}
The identity for $dC/d\tau$ is analogous.
\end{proof}

A direct corollary of Lemma \ref{lem:sensitivity} is the \emph{frontier slope identity}:
\begin{align}
    \frac{dQ}{dC}(\tau)
    = \frac{\Delta Q(\tau)}{\Delta C(\tau)}
    = \rho(\tau),
    \label{eq:slope}
\end{align}
i.e., the slope of the achievable $(Q, C)$ frontier at fallback threshold $\tau$ equals the local benefit–cost ratio of requests with score $s=\tau$.
\begin{assumption}[Informativeness of the score]
    The cloud-edge marginal cost differential $\Delta C(s)$ is strictly positive on its support; $\Delta Q(\cdot)$ and $\Delta C(\cdot)$ are continuous with the ratio $\rho(s)\triangleq \Delta Q(s)/\Delta C(s)$ strictly decreasing in $s$.
    \label{ass:cost}
\end{assumption}
Assumption \ref{ass:cost} is natural: as the confidence score increases, the cloud's incremental benefit relative to the edge diminishes faster than its incremental cost, making a single switch optimal. 

We next solve the scalar optimization $\max_{\tau} J(\tau)$ and show the uniqueness of the optimal fallback threshold $\tau^{\ast}$.

\begin{theorem}[Unique optimal fallback threshold (first-order balance)]
\label{thm:threshold}
Under Assumption~\ref{ass:regularity} and Assumption~\ref{ass:cost}, the score-threshold policy is optimal among all measurable policies that depend only on the score, and the maximizer $\tau^*$ of $J(\tau)$
is unique. It is characterized by the first-order balance
\begin{align}
    \Delta Q\!\big(\tau^*\big) = \lambda\, \Delta C\!\big(\tau^*\big) , \qquad
    \operatorname{sign}\!\big(J'(\tau)\big)
    = \operatorname{sign}\!\big(\rho(\tau)-\lambda\big).
    \label{eq:firstorder}
\end{align}
\end{theorem}

\begin{proof}
By Lemma~\ref{lem:sensitivity},
\begin{align}
    J'(\tau)
    &= \frac{dQ}{d\tau}(\tau) - \lambda \frac{dC}{d\tau}(\tau) \notag \\
    &= f(\tau)\!\left( \Delta Q(\tau) - \lambda \Delta C(\tau) \right)  \\
    &= f(\tau)\Delta C(\tau)\!\left( \rho(\tau)-\lambda \right).\notag
\end{align}
Assumption~\ref{ass:regularity} and Assumption~\ref{ass:cost} ensure $f(\tau)>0$ and $\Delta C(\tau)>0$, respectively. 
Hence, the sign of $J'(\tau)$ is the sign of $\rho(\tau)-\lambda$. 
Because $\rho(\cdot)$ is strictly decreasing and continuous, 
there exists a unique $\tau^*$ satisfying $\rho(\tau^*)=\lambda$; 
at that point $J'(\tau^*)=0$, and strict monotonicity of $\rho$ implies $J'$ changes sign, yielding a unique maximizer $\tau^*$. 
Optimality of the fallback threshold structure among score-based policies follows from the single-crossing of $\rho(\cdot)$. 
Finally, by Lemma~\ref{lem:sensitivity}, the frontier slope at $\tau$ equals $dQ/dC=\rho(\tau)$, so it equals $\lambda$ at $\tau^*$.
\end{proof}

\begin{remark}
Under the mild regularity conditions of Theorem~\ref{thm:threshold}, the Pareto frontier traced by $\tau \mapsto (Q(\tau),C(\tau))$ is strictly monotone and differentiable. Meanwhile, the optimal routing policy possesses three interpretable properties:
(i) the optimal measurable score-based policy is threshold-shaped;
(ii) the fallback threshold $\tau^{\ast}$ is unique and satisfies a first-order balance
between the marginal expected quality gain and the marginal offload cost;
and (iii) 
at the unique maximizer $\tau^*$, its tangent slope equals $\lambda$.
\end{remark}

This result provides a principled basis for incorporating network state $S$ (e.g., bandwidth, RTT) into the fallback threshold $\tau$, which we extend to the state-aware case below.
Our goal is to characterize how the optimizer 
$\tau^*(S)$ 
of 
$J(\tau;S)=Q(\tau;S)-\lambda C(\tau;S)$ 
shifts with $S$.

\begin{assumption}[Separable network impact]
\label{ass:network}
There exists a strictly positive, continuous scaling $\kappa(S)$ such that 
$\Delta C_S(s)=\kappa(S)\,\Delta C(s)$ for all $s$; 
moreover, the marginal quality gain is governed by task difficulty captured by the score and is not distorted by the link, i.e., $\Delta Q_S(s)=\Delta Q(s)$.
\end{assumption}

By Assumption \ref{ass:network} and Theorem \ref{thm:threshold}, the optimal fallback threshold $\tau^*(S)$ is characterized by
\begin{align}
\Delta Q\!\big(\tau^*(S)\big)
= \lambda\,\Delta C_S\!\big(\tau^*(S)\big)
= \lambda\,\kappa(S)\,\Delta C\!\big(\tau^*(S)\big), 
\label{eq:network-optimality}
\end{align}
or equivalently $\rho\!\big(\tau^*(S)\big)=\lambda\,\kappa(S)$. 
\begin{theorem}[Network influence on $\tau^*$]
\label{thm:monotone}
Under Assumption \ref{ass:cost} and Assumption~\ref{ass:network}, 
if $\kappa(S_1)>\kappa(S_2)$ then $\tau^*(S_1)<\tau^*(S_2)$. 
\end{theorem}

\begin{proof}
Start from the first-order condition at state $S$:
\begin{align}
\rho\!\left(\tau^*(S)\right)=\lambda\,\kappa(S).
\end{align}
Pick $S_1,S_2$ with $\kappa(S_1)>\kappa(S_2)$. Then
\begin{align}
\lambda\,\kappa(S_1)>\lambda\,\kappa(S_2).
\end{align}
By strict monotonic decrease of $\rho(\cdot)$ (Assumption~\ref{ass:cost}), 
the equation $\rho(\tau)=\lambda\,\kappa(S)$ has a unique solution and the solution moves in the \emph{opposite} direction of the right-hand side. 
Hence
\begin{align}
\rho\!\big(\tau^*(S_1)\big)>\rho\!\big(\tau^*(S_2)\big)
\quad\Longrightarrow\quad
\tau^*(S_1)<\tau^*(S_2),
\end{align}
which proves the theorem. 
\end{proof}
\begin{remark}
Theorem \ref{thm:monotone} implies that if $\kappa$ is strictly decreasing in available bandwidth $\mathrm{BW}$ and strictly increasing in RTT or cloud-unit price $C_{\mathrm{tok}}$, then $\tau^*$ is strictly increasing in $\mathrm{BW}$ and strictly decreasing in RTT and $C_{\mathrm{tok}}$.
\end{remark}
\begin{corollary}[Local sensitivity of $\tau^*(S)$]
\label{cor:sensitivity-network}
Suppose $\rho$ is continuously differentiable and $\rho'(\tau) \triangleq \frac{d \rho}{d \tau}<0$. 
Then $\tau^*$ is continuously differentiable in $\kappa$, and
\begin{align}
\frac{d\tau^*}{d\kappa}(S)
= \frac{\lambda}{\rho'\!\big(\tau^*(S)\big)} < 0.
\label{eq:implicit-derivative}
\end{align}
\end{corollary}

\begin{proof}
Define the residual:
\begin{align}
F(\tau,\kappa)\triangleq\rho(\tau)-\lambda\,\kappa.
\end{align}
At $\big(\tau^*(S),\kappa(S)\big)$, the optimality condition gives 
$F\!\big(\tau^*(S),\kappa(S)\big)=F_\tau(\tau^*,\kappa)\,d\tau^* + F_\kappa(\tau^*,\kappa)\,d\kappa=0$, which implies:
\begin{align}
\frac{d\tau^*}{d\kappa} &= -\frac{F_\kappa}{F_\tau}
= -\frac{-\lambda}{\rho'(\tau^*)}
= \frac{\lambda}{\rho'(\tau^*)}.
\end{align}
This establishes the corollary.
\end{proof}

Finally, the frontier's local geometry inherits the same parameterization: 
for each fixed $S$ and fallback threshold $\tau$,
\begin{align}
\frac{dQ}{dC}(\tau;S)
= \frac{\Delta Q(\tau)}{\Delta C_S(\tau)}
= \frac{\rho(\tau)}{\kappa(S)}.
\label{eq:frontier-network}
\end{align}
Therefore, improving network states (smaller $\kappa(S)$) steepens the attainable quality-cost slope at the same fallback threshold. Together with the first-order balance $\rho(\tau^{\ast}(S)) = \lambda\,\kappa(S)$, this shows that the network state shifts the optimal fallback threshold monotonically. 
In practice, this motivates either a linearized functional mapping (FuncDyn) or a lightweight one-hidden-layer router (PolicyNet).

\begin{itemize}
  \item \emph{FuncDyn}:
It computes the fallback threshold $\tau_{i,k}$ as a smooth function of the observed link condition and recent execution history:
\begin{align}
\tau_{i,k}=\tau_0-\alpha_{\mathrm{RTT}}\mathrm{RTT}_{i,k}+\beta_{\mathrm{BW}}\mathrm{BW}_{i,k}-\gamma_{\mathrm{hist}}\widehat{Q}_{i,k},
\label{eq:dy_1}
\end{align}
Here $\tau_0$ denotes the base fallback threshold level obtained from the initialization procedure described in Appendix A, and serves as the reference point adjusted by the observed network state. $\widehat{Q}_{i,k}$ denotes a running estimate of $Q_{i,k}$.

\begin{figure}
	\centering
	\includegraphics[width=0.495\textwidth]{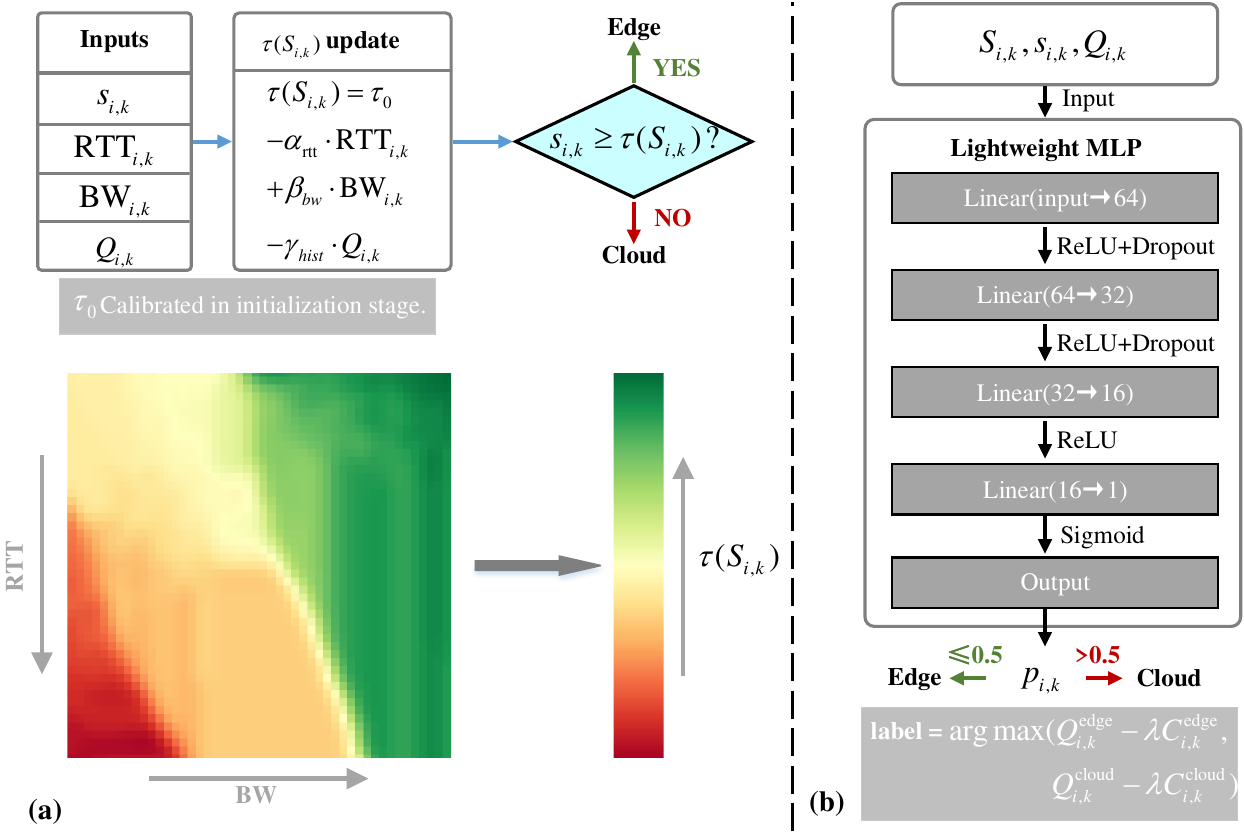}
	\caption{Runtime routing modules and adaptive fallback thresholding. (a) Dynamic fallback threshold controller $\tau_{i,k}=f_{\phi}(S_{i,k},\widehat{Q}_{i,k})$, initialized from the offline-calibrated $\tau_0$ and adjusted by the current network state $(\mathrm{RTT}_{i,k}, \mathrm{BW}_{i,k})$ and recent quality $\widehat{Q}_{i,k}$. The heatmap shows how network state variables shape $\tau_{i,k}$, with $\widehat{Q}_{i,k}$ fixed at its validation mean. (b) Lightweight router $f_{\omega}$ (PolicyNet) takes $[S_{i,k}, s_{i,k}, \widehat{Q}_{i,k}]$ and outputs routing probability $p_{i,k}$; decisions follow $p_{i,k} \le 0.5 \Rightarrow \text{Edge}$, $p_{i,k} > 0.5 \Rightarrow \text{Cloud}$.}
	\label{fig:Routing}
\end{figure}

As shown in Fig.~\ref{fig:Routing}(a), $\tau_{i,k}$ increases with bandwidth and decreases with RTT, consistent with the comparative statics established in Theorem~\ref{thm:monotone}.
\item \emph{PolicyNet}:
It eliminates the explicit fallback threshold and learns a small neural policy that directly outputs the routing decision.
Given the feature tuple $(S_{i,k}, s_{i,k}, \widehat{Q}_{i,k})$, the PolicyNet predicts the probability of routing to the cloud:
\begin{align}
p_{i,k} &= \sigma\!\big(f_{\omega}(S_{i,k}, s_{i,k}, \widehat{Q}_{i,k})\big) \ge 0.5, 
\label{eq:dy_2}
\end{align}
where the neural network $f_{\omega}(\cdot)$ is a compact MLP ($4$ layers, $\approx 2,900$ parameters), ensuring it can be deployed efficiently on the edge device.
The neural network $f_{\omega}(\cdot)$ can be trained following a standard logistic objective,
\begin{align}
\mathcal{L}_{\text{policy}}(\omega)
= \mathbb{E}\!\left[-y_{i,k}\log p_{i,k}-(1-y_{i,k})\log(1-p_{i,k})\right], 
\end{align}
where the training label $y_{i,k}
= \operatorname*{arg\,max}\big(Q^{\mathrm{edge}}_{i,k}-\lambda C^{\mathrm{edge}}_{i,k},
\, Q^{\mathrm{cloud}}_{i,k}-\lambda C^{\mathrm{cloud}}_{i,k}\big)$.
\end{itemize}
\subsection{Online Self-Improvement}
After initialization, the edge agent continues to refine itself (i.e., $\pi_\theta$), contingent on the unified score signal $s_{i,k}$ and shared cached edge–cloud pairs collected during deployment\footnote{We leave the details of caching setup in Appendix \ref{sec:log}.}. 
At step $(i,k)$, the router score acts as the immediate reward by 
\begin{align}
r_{i,k} = s_{i,k} = g_\psi(x_i, H_{i,k}, u_{i,k}). 
\label{eq:r=s}
\end{align}
The goal is to maximize the expected discounted return in 
\begin{align}
\max_{\theta}\;
\mathbb{E}_{\pi_\theta}\!\left[\sum_{k}\gamma^{k} r_{i,k}\right],
\qquad 0 < \gamma \le 1. 
\label{eq:reward_target}
\end{align}
The discount factor $\gamma$ controls how strongly future rewards affect current updates.

Commonly, we follow
\begin{align}
\textbf{(RL-KL)}\quad
\widetilde{\pi}_{t+1}
&=\arg\max_{\pi}\;
\underbrace{\mathbb{E}_{(s,a)\sim d_{\pi_t}}\!\big[A_{\pi_t}(s,a)\,\log \pi(a|s)\big]}_{\text{policy improvement}} \notag \\
&\;-\;\eta_t\,\mathrm{KL}\!\left(\pi\,\|\,\pi_t\right).
\end{align}
and apply a clipped PPO update \cite{Schulman2017PPO} by optimizing:
\begin{align}
\mathcal{L}_{\mathrm{RL}}(\theta)
&= \mathbb{E}\!\Big[\min(\rho_{i,k} A_{i,k},
\,\mathrm{clip}(\rho_{i,k}, 1 - \epsilon, 1 + \epsilon) A_{i,k})\Big]\notag \\
&- \beta_{\mathrm{KL}} \mathrm{KL}\!\big(\pi_\theta \Vert \pi_{\mathrm{old}}\big),
\label{eq:PPO}
\end{align}
where
\begin{align}
\rho_{i,k} = 
\frac{\pi_\theta(u_{i,k}\mid x_i,H_{i,k})}
{\pi_{\mathrm{old}}(u_{i,k}\mid x_i,H_{i,k})}, \notag
\end{align}
$A_{i,k}$ is the advantage estimate from rewards, $\epsilon$ indicates the clipping range, and $\beta_{\mathrm{KL}}$ is the KL penalty coefficient controlling the update size. 

On the other hand, to maintain structured outputs, every $M$ PPO updates are followed by a short SFT anchoring step:
\begin{align}
\textbf{(SFT-CE)}\quad
\pi_{t+1}
=\arg\min_{\pi}\;
\mathbb{E}_{s\sim \mathcal{D}_{\text{SFT}}}
\Big[
\underbrace{
H\big(\pi_{\text{SFT}}(\cdot|s),\,\pi(\cdot|s)\big)
}_{=\;\mathrm{KL}(\pi_{\text{SFT}}\,\|\,\pi)+\text{const}}
\Big].
\label{eq:SFT_a}
\end{align}
where $\pi_{\mathrm{SFT}}$ is the frozen supervised model.
This projection keeps the tool schema correct while allowing the edge LLM policy to improve.
Over time, $\pi_\theta$ learns to approach the cloud's performance with reduced cost.


The first stage (RL-KL) performs a KL-regularized trust-region update that delivers stable policy improvement for a regularized MDP \cite{geist2019theory}. Meanwhile, the second stage (SFT-CE) applies an SFT-anchored cross-entropy projection that pulls the policy back toward the SFT manifold to preserve formatting/tool-calling habits \cite{banerjee2005clustering}. Jointly, the alternation can be viewed as block-coordinate ascent on a composite regularized objective:
\begin{align}
\begin{aligned}
J_{\eta,\mu}(\pi;\pi_t,\pi_{\text{SFT}})
&= \mathbb{E}_{\,s\sim d_{\pi_t},\,a\sim \pi(\cdot|s)}[A(s,a)] \\
&\quad-\;\eta\,\mathbb{E}_{\,s\sim d_{\pi_t}}[\mathrm{KL}(\pi(\cdot|s)\,\|\,\pi_t(\cdot|s))] \\
&\quad-\;\mu\,\mathbb{E}_{\,s\sim d_{\pi_t}}[\mathrm{KL}(\pi_{\text{SFT}}(\cdot|s)\,\|\,\pi(\cdot|s))].
\end{aligned}
\label{eq:rl_obj}
\end{align}
Here $A(s, a)$ is the same advantage estimator as defined in Eq.~(\ref{eq:PPO}), and $d_{\pi_t}$ denotes the discounted state distribution under the current policy.


Similar to the initialization method in Appendix \ref{sec:initialization}, the reward model is incrementally updated using the same pairwise ranking loss with schema penalty as Eq.~(\ref{eq:RM_train}).
Fine-tuning is performed in small batches with early stopping to prevent reward drift. This incremental refresh keeps $g_\psi$ aligned with the evolving edge LLM policy $\pi_\theta$ and the time-varying network state $S$, ensuring that the router signal $s_{i,k}$ remains consistent and informative.

Finally, we summarize the whole procedure for the cloud-edge synergy in Algorithm \ref{alg:routing_online}.
\begin{algorithm}[t]
\caption{Network-Aware Routing and Online Self-Improvement}
\label{alg:routing_online}
\textbf{Inputs:} dataset $\mathcal{D}$, schema $\mathcal{S}$, trade-off $\lambda$, link logs, \\
\hspace*{11mm} edge LLM policy $\pi_{\mathrm{SFT}}$, reward model $g_{\psi}$, horizon $M$.\\
\textbf{Outputs:} updated $\pi_{\theta}$, $\tau(\cdot)$ or $f_{\omega}$, refreshed $g_{\psi}$.
\vspace{3pt}
\hrule
\vspace{3pt}

\textbf{Offline Initialization}
\begin{algorithmic}[1]
\State Obtain $\pi_{\mathrm{SFT}}$ by Eq.~(\ref{eq:SFT}) 
\State Obtain $g_{\psi}$ by Eq.~(\ref{eq:RM_train})
\State Obtain $\tau_0$ by Eq.~(\ref{eq:initial_tau})
\State Set $\pi_{\theta}\!\leftarrow\!\pi_{\mathrm{SFT}}$
\end{algorithmic}

\vspace{4pt}
\hrule
\vspace{4pt}

\textbf{Online Routing and Self-Improvement}
\begin{algorithmic}[1]
\For{each task $i$}
  \State $k\!\leftarrow\!1$, context $H_{i,1}$
  \While{task $i$ not finished}
     \State $u_{i,k}\!\sim\!\pi_{\theta}(\cdot\!\mid\!x_i,H_{i,k})$
     \State Obtain $s_{i,k}=g_{\psi}(x_i,H_{i,k},u_{i,k})$ by Eq.~(\ref{eq:reward})
     \State Obtain $\tau_{i,k}$ by Eq.~(\ref{eq:dy_1}) or Obtain $p_{i,k}$ by Eq.~(\ref{eq:dy_2})
     \If{$s_{i,k}\!\ge\!\tau_{i,k}$ \text{ or } $p_{i,k}\ge0.5$}
         \State Execute $u_{i,k}$ locally
     \Else
         \State Query cloud $\Rightarrow u_{i,k}^{\star}$
         \State Obtain $C_{i,k}$ by Eq.~(\ref{eq:cloud_cost})
         \State Store $(x_i,H_{i,k},u_{i,k},u_{i,k}^{\star},s_{i,k},S_{i,k})$ in $\mathcal{B}_{\mathrm{RM}}$
     \EndIf
    
     \State Compute $Q_{i,k}$ by Eq.~(\ref{eq:quality})
     \State Compute $L_{i,k}$ by Eq.~(\ref{eq:latency})
     \State Update history $\widehat Q_{i,k}\leftarrow(1-\beta)\widehat Q_{i,k-1}+\beta Q_{i,k}$
     \State Update $H_{i,k+1}$ by Eq.~(\ref{eq:update})
     \State Add $(x_i,H_{i,k},u_{i,k},r_{i,k}{=}s_{i,k})$ to $\mathcal{B}_{\mathrm{RL}}$
     \State $k\!\leftarrow\!k{+}1$
  \EndWhile
\EndFor
\Statex
\For{each idle window}
    \State Update $\pi_{\theta}$ by PPO using Eq.~(\ref{eq:PPO}) on $\mathcal{B}_{\mathrm{RL}}$
    \If{every $M$ steps}
        \State Align $\pi_{\theta}$ to $\pi_{\mathrm{SFT}}$ by Eq.~(\ref{eq:SFT_a})
    \EndIf
    \State Update $g_{\psi}$ using incremental RM training (Eq.~(\ref{eq:reward})) on $\mathcal{B}_{\mathrm{RM}}$
    \State Recalibrate $\tau_{i,k}$ or $f_{\omega}$ using latest $(S_{i,k},\widehat Q_{i,k})$
\EndFor
\State \Return $\pi_{\theta},\ \tau(\cdot)\text{ or }f_{\omega},\ g_{\psi}$
\end{algorithmic}
\end{algorithm}

\section{Performance Evaluation}
\label{sec:evaluation}

\subsection{Experimental Setup}

We consider a NetGPT scenario with a heterogeneous tool-calling stack where the edge runs \texttt{DeepSeek-R1-Distill-Qwen-7B} (FP32, non-quantized) and the cloud uses \texttt{DeepSeek-V3.2-Exp}. 
Both models are required to emit structured tool calls with a tool name, arguments, and a brief thought trace, but we use different prompting schemes at the edge and in the cloud. At the edge, we use a compact schema that includes the current query, tool names from previous steps, and the list of currently available tools with their argument slots. In the cloud, we use a ReAct-style prompt \cite{Yao2023ReAct} with richer task instructions and full tool and argument descriptions to enforce the desired reasoning pattern and output format.

\begin{figure}
	\centering
	\includegraphics[width=0.495\textwidth]{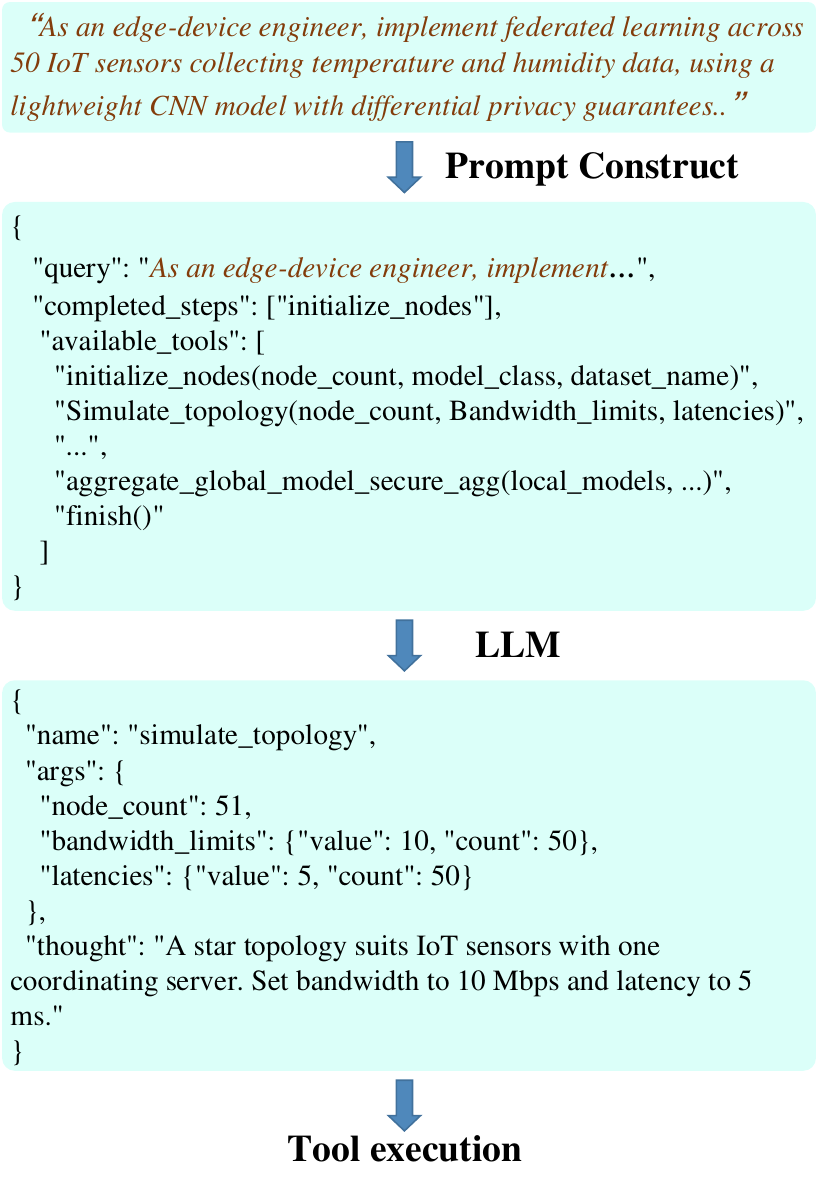}
	\caption{Example of structured text flow across stages.}
	\label{fig:json_example}
\end{figure}

At the edge, we optimize the local LLM with PPO \cite{Schulman2017PPO}, using the unified router score as a scalar reward. The score $s_{i,k}$ is produced by a lightweight reward model $g_{\psi}$ instantiated on \emph{Qwen2.5-1.5B-Instruct}, fine-tuned on comparison judgments favor cloud-level tool-calling quality and to penalize schema or format violations. During deployment, $g_{\psi}$ serves as the scoring module of the router, while the routing policy leverages $s_{i,k}$, the current network state $S_{i,k}$, and the running quality estimate $\widehat{Q}_{i,k}$ to decide whether to continue execution on the edge or offload the next step to the cloud.

Workloads are drawn from a private corpus of tool-calling tasks that matches our target skill mix. From this corpus\footnote{As described in Section \ref{sec:method} and Appendix \ref{sec:log}, only cloud-offloaded steps are logged and reused.}, we sample $8,000$ examples for SFT-based initialization of the edge LLM for stabilizing schema-compliant tool-calling, $2,000$ examples to train the reward model, and $8,000$ tasks for inference experiments and simulated continual improvement. Task difficulty scales with input length and tool-calling complexity: each instance exposes $10-20$ candidate tools and includes between $0$ and $8$ previously executed steps. During evaluation, requests are processed step-wise and metrics are aggregated at the task level.
To couple routing decisions with network conditions, each request samples a throughput $\mathrm{BW}$ (Mbps) and a round-trip time $\mathrm{RTT}$ (ms) from one of three regimes:
\begin{itemize}
    \item GOOD: $\mathrm{BW} \in [120,200]$, $\mathrm{RTT} \in [20,40]$;
    \item MID: $\mathrm{BW} \in [30,80]$, $\mathrm{RTT} \in [40,80]$;
    \item BAD: $\mathrm{BW} \in [5,15]$, $\mathrm{RTT} \in [80,130]$.
\end{itemize}
This design aligns our three buckets with established QoS envelopes: according to the ITU-T Y.1541 \cite{ITU-T-Y1541-2011} end-to-end objectives, $20-40$ ms RTT reflects metro/near-region paths, while $40-80$ ms RTT corresponds to typical regional paths, and $80-130$ ms RTT indicates long-haul yet standards-consistent paths. For bandwidth, we sample a link rate in Mbps as a controlled variable but interpret its effect through the TCP bandwidth-delay-product (BDP) rationale (RFC 6349) \cite{rfc6349}, i.e., GOOD/MID/BAD correspond to low-/moderate-/high-BDP regimes that govern achievable throughput and window sizing. We model time-varying network state via a one-step Gauss-Markov update $S_{i,k} = S_{i,k-1} + \varepsilon_{i,k}$, where $\varepsilon_{i,k} \sim \mathcal N(0,\Sigma^2)$. For time-correlated links, such a method underpins widely used network/wireless mobility abstractions \cite{GHANDOURHAIDAR2012601}.

\begin{figure*}[t]
  \centering
  \includegraphics[width=0.98\textwidth]{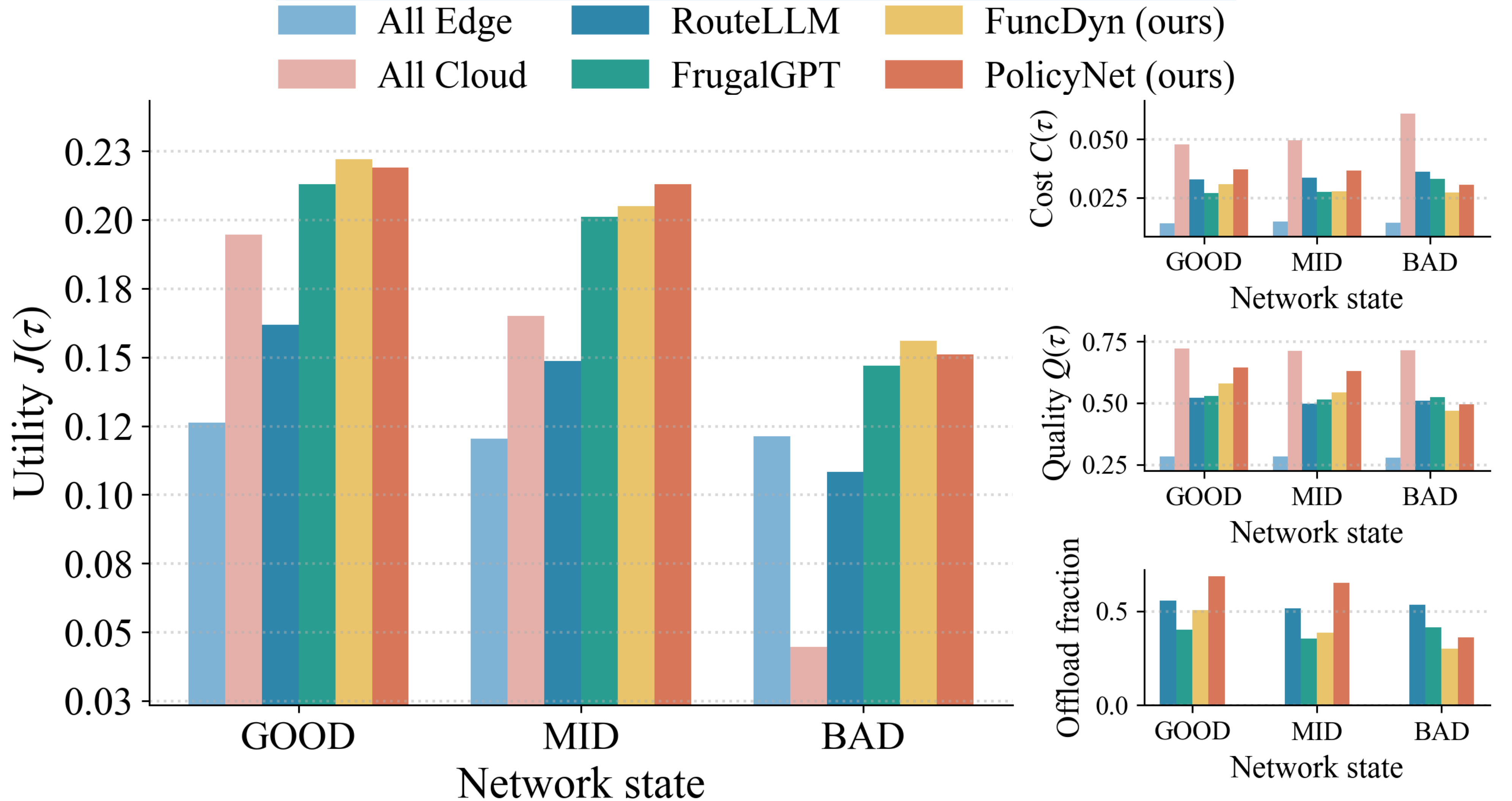}
  \caption{Comparison among All Edge, All Cloud, RouteLLM \cite{ong-etal-2025-routellm}, FrugalGPT \cite{chen-etal-2024-frugalgpt}, and the dynamic controllers (FuncDyn, PolicyNet) under GOOD/MID/BAD links.}
  \label{fig:controller-comparison}
\end{figure*}

\begin{figure*}[t]
  \centering
  \includegraphics[width=0.98\textwidth]{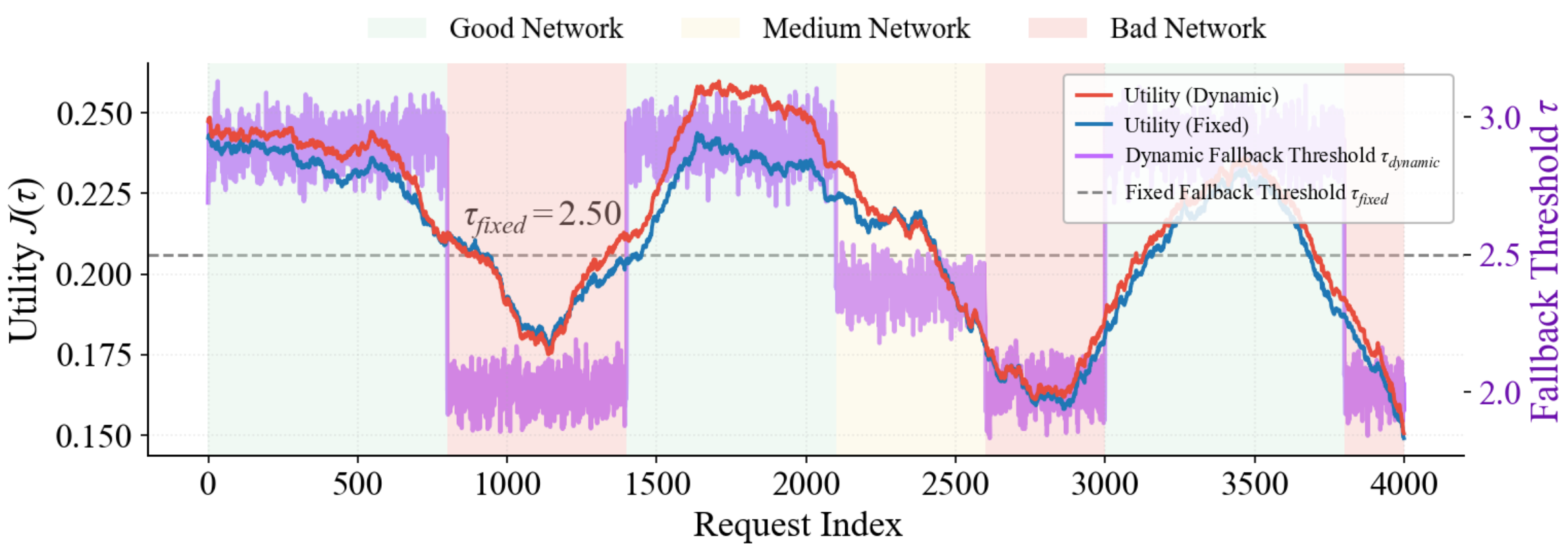}
  \caption{Dynamic vs.\ fixed fallback threshold under time-varying links.
  }
  \label{fig:dynamic-threshold}
\end{figure*}

We study two policies: \emph{FuncDyn} constructs a rule-based dynamic fallback threshold $\tau_{i,k}$ from the normalized round-trip time $\mathrm{RTT}_{i,k}$, bandwidth $\mathrm{BW}_{i,k}$, and a historical quality $\widehat{Q}_{i,k}$, and offloads when $s_{i,k} > \tau(\mathrm{RTT}_{i,k}, \mathrm{BW}_{i,k}, \widehat{Q}_{i,k})$. \emph{PolicyNet} is a compact four-layer MLP that takes as input $x_{i,k} = \bigl(S_{i,k}, s_{i,k}, \widehat{Q}_{i,k}\bigr)$ and outputs a binary cloud-edge routing decision. For comparison, we evaluate two routing baselines. \emph{RouteLLM}~\cite{ong-etal-2025-routellm} is a learned one-shot router that, given an input query, makes a single binary choice—answer on the edge or offload to a stronger model—and keeps this choice fixed throughout multi-step tool calling. \emph{FrugalGPT}~\cite{chen-etal-2024-frugalgpt} implements a fixed fallback threshold cascade, where requests are passed from smaller to larger models according to a static acceptance rule.
All training and batched inference run on $8\times$ NVIDIA A800 (80 GB) GPUs. 

Metrics include \emph{utility} $J$, \emph{composite quality} $Q$, \emph{total cost} $C$, and \emph{offload rate} (fraction of requests routed to the cloud). Unless otherwise stated, all metrics are reported as averages over the $8,000$ evaluation tasks. 
We record latency components but report normalized costs so that $Q$ and $C$ remain commensurate under the chosen $\lambda$. In addition, for frontiers and sensitivity analyses, we sweep $\tau$ over a dense grid; for long traces, we interleave regimes to emulate time-varying links.

\subsection{Adaptive Routing under Time-Varying Links}

We compare our dynamic controllers (FuncDyn and PolicyNet) with RouteLLM, FrugalGPT, and two additional baselines (All Edge and All Cloud) across GOOD/MID/BAD regimes, and Fig.~\ref{fig:controller-comparison} presents the corresponding results. It shows FuncDyn yields the highest $J$ while keeping the offload rate moderate, reflecting effective state-aware acceptance. PolicyNet attains the highest $Q$ with higher offload and cost. In contrast, FrugalGPT ranks between the dynamic controllers and RouteLLM across regimes. Under regime switches, RouteLLM yields the lowest utility $J$, while our dynamic controllers remain top-performing. Together with the baselines All Edge and All Cloud, the comparison indicates that coupling the router with both execution agents yields a larger utility $J$ across GOOD/MID/BAD. 

We evaluate routing under interleaved network conditions by tracking the per-request utility $(J=Q-\lambda C)$ over a long trace. Fig.~\ref{fig:dynamic-threshold} plots the step-level $J$ together with its windowed mean for a dynamic, network-aware fallback threshold and for a fixed fallback threshold representative of a static acceptance policy. Across regime switches, the dynamic fallback threshold maintains a consistently higher average $J$: when the link deteriorates, the controller lowers the acceptance boundary to curb offloading cost; as the link recovers, the boundary lifts to exploit cloud quality gains. The utility series shows smooth piecewise transitions without overshoot or oscillations, indicating stable adaptation under non-stationary links.

\subsection{Quality-Cost Frontier and Network Sensitivity}

Fig.~\ref{fig:frontier} reports the attainable quality-cost frontiers by sweeping the fallback threshold $\tau$ under three fixed network conditions (GOOD, MID, BAD). The curves are smooth, monotonic in cost, and show diminishing marginal quality gains as cost increases, empirically supporting the decreasing-benefit Assumption~\ref{ass:cost} used in our analysis. Consistent with Eq.~(\ref{eq:frontier-network}), at the same fallback threshold, the frontier slope scales as $1/\kappa(S)$, so degraded links (larger $\kappa$) flatten the curve and shift it rightward/downward relative to better network states. Across regimes, the frontier shifts rightward (and slightly downward) from GOOD to BAD, indicating that worse links require more cost to reach the same quality and deliver lower quality at the same cost—consistent with our theory that degraded networks raise the effective offloading cost and push decisions toward edge execution. 

\begin{figure}[t]
    \centering
    \includegraphics[width=0.95\linewidth]{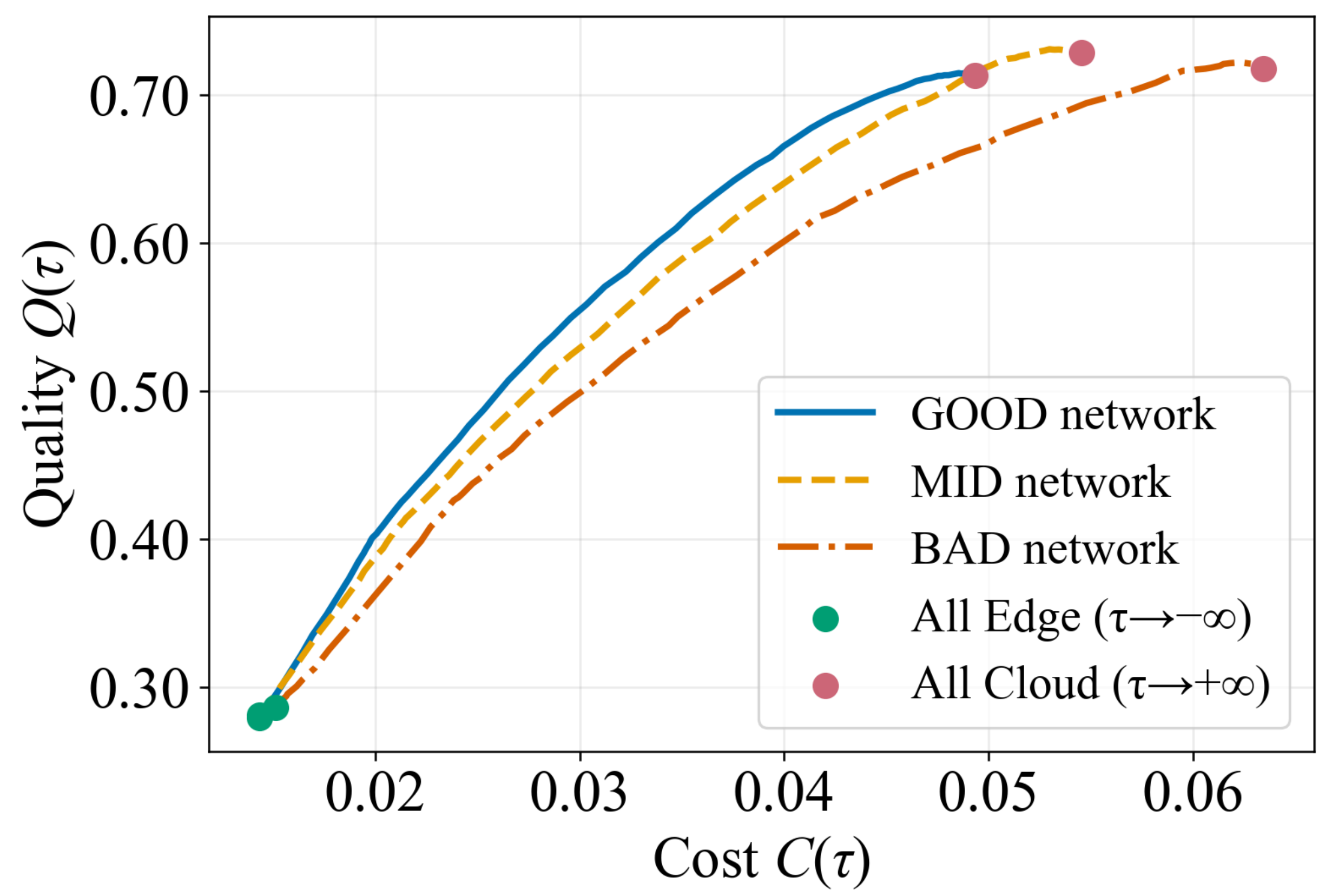}
    \caption{Quality-Cost frontiers across fixed network regimes.}
    \label{fig:frontier}
\end{figure}
\begin{figure}[t]
  \centering
  \includegraphics[width=0.95\linewidth]{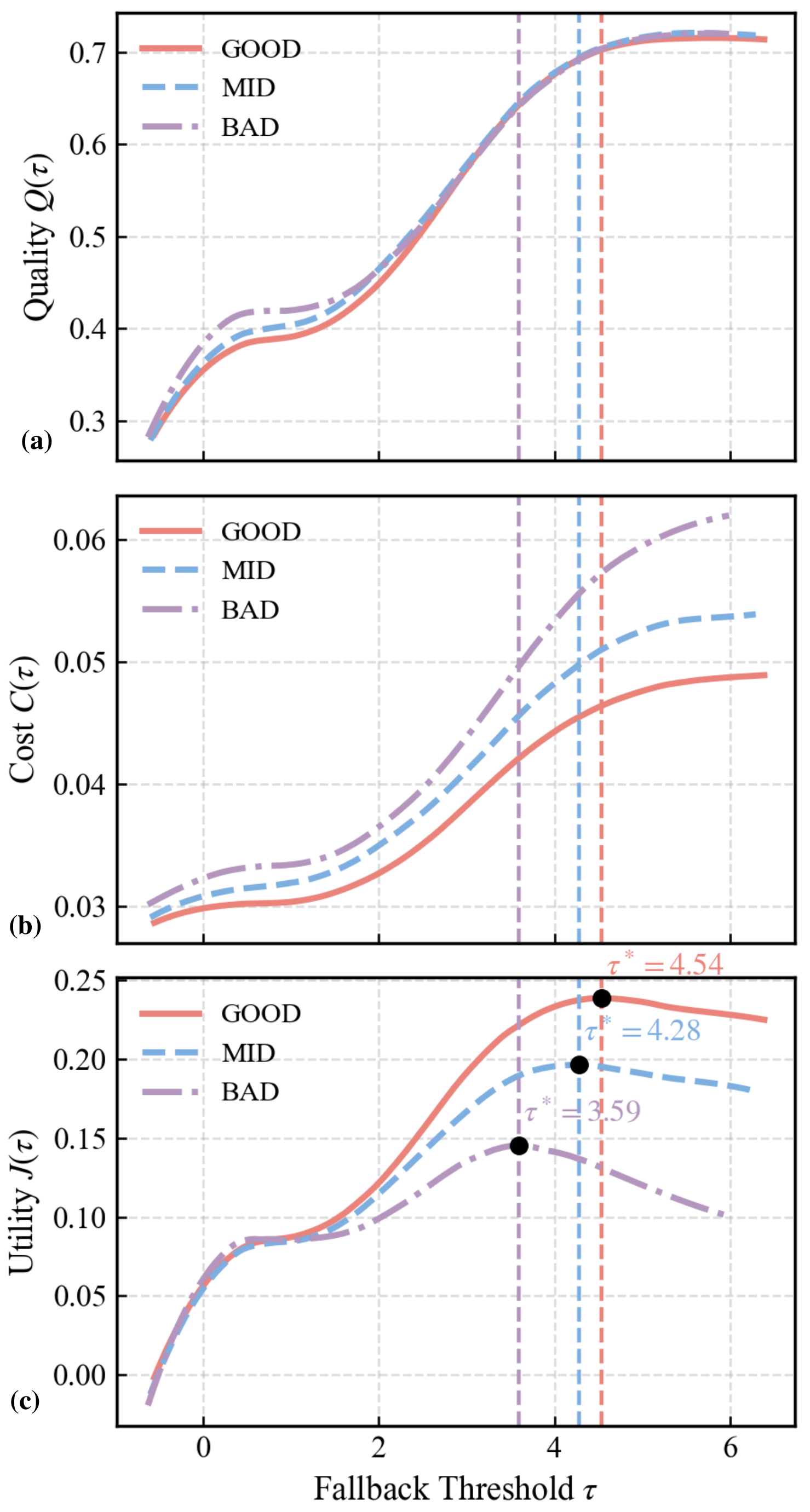}
  \caption{Network sensitivity of the score-threshold policy. 
  (a)~Composite quality~$Q(\tau)$, 
  (b)~total cost~$C(\tau)$, and 
  (c)~utility~$J(\tau)=Q-\lambda C$ 
  under the three network states.}
  \label{fig:sensitivity}
\end{figure}
\begin{figure}[t]
  \centering
  \includegraphics[width=0.98\linewidth]{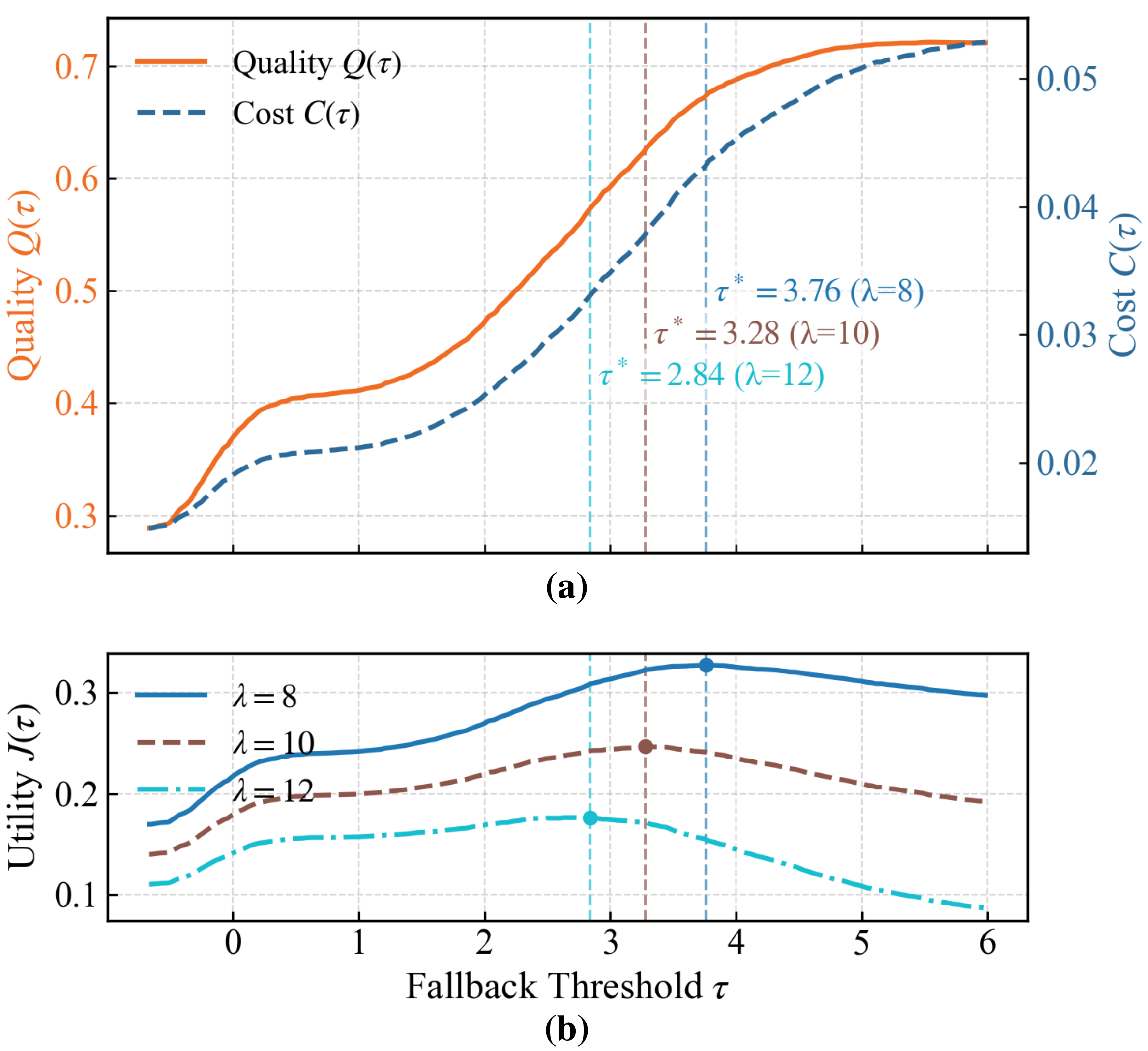}
  \caption{Fallback threshold scan: $Q(\tau)$, $C(\tau)$, and $J(\tau)$ for $\lambda\in\{8,10,12\}$; vertical lines indicate $\tau^*$ for each $\lambda$.}
  \label{fig:threshold-scan}
\end{figure}

Fig.~\ref{fig:sensitivity} shows the sensitivity of the score-threshold policy to network conditions. We sweep the fallback threshold $\tau$ under three network regimes while fixing the trade-off weight at $\lambda = 10$, and plot the resulting quality, cost, and utility curves. As RTT increases and bandwidth decreases from GOOD to MID to BAD, both curves shift, and the utility-maximizing fallback threshold moves leftward. The empirical optima $\tau^{*}_{\text{GOOD}}\approx 4.54$, $\tau^{*}_{\text{MID}}\approx 4.28$, and $\tau^{*}_{\text{BAD}}\approx 3.59$ follow the monotone comparative-statics pattern predicted by our analysis: deteriorating links impose higher communication costs, which discourage cloud offloading and favor more permissive local acceptance at the edge. These results confirm that the proposed score-threshold policy adapts smoothly and predictably as link quality worsens.

For completeness, Fig.~\ref{fig:threshold-scan} further decomposes the fallback threshold scan by plotting $Q(\tau)$ and $C(\tau)$ together with overlaid utility curves $J(\tau)$ for different values of $\lambda$.
Across all settings, the $J(\tau)$ profiles remain unimodal, and the maximizing fallback threshold shifts leftward as $\lambda$ increases, reflecting heightened sensitivity to cost.

\subsection{Effectiveness of PPO-based Learning}
We assess the effectiveness of on-device policy optimization with the SFT anchor to preserve schema-correct tool-calling. Fig.~\ref{fig:ppo-training} contrasts training with and without the SFT anchor. It can be observed that the involvement of SFT leads to monotonically increased reward with lower variance and a steadily declining offload rate, indicating reliable improvement without structural drift. These results corroborate that interleaved SFT updates stabilize RL by constraining updates around a schema-preserving reference. Fig.~\ref{fig:ppo-frontier} shows the quality-cost frontier before and after PPO training. After training, the frontier strictly dominates the pre-training curve across a range of $\lambda$ values, suggesting that the utility gains are robust and not tied to a single trade-off weight.
From a multi-agent perspective, the router, the edge LLM policy, and the cloud model operate in a network feedback loop. The improvements in Fig.~\ref{fig:ppo-training} and the shift of the quality-cost frontier in Fig.~\ref{fig:ppo-frontier} are consistent with coordinated updates driven by online feedback rather than offline tuning alone. 

\begin{figure}[t]
  \centering
  \includegraphics[width=0.95\linewidth]{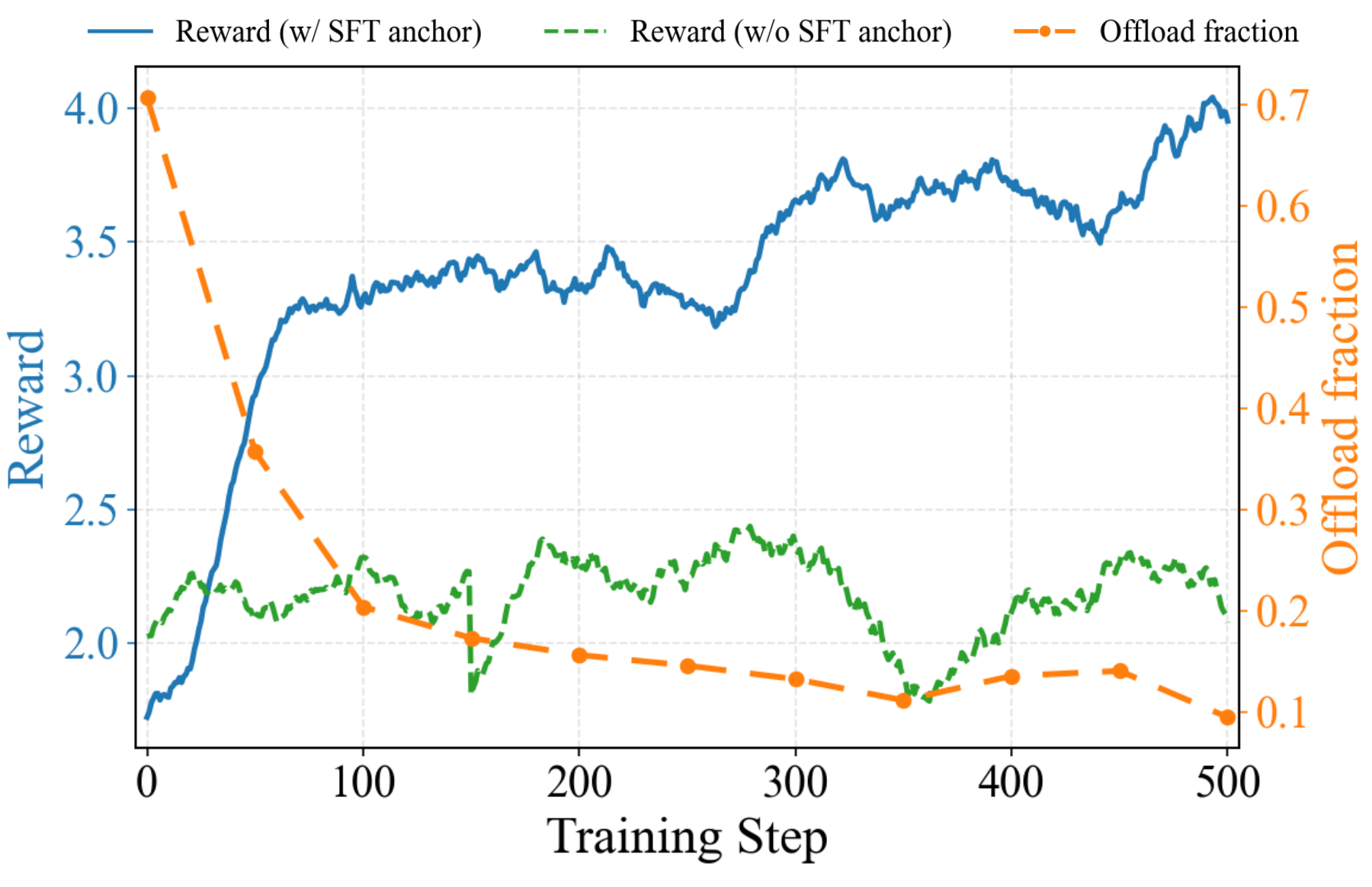}
  \caption{PPO with vs. without SFT anchoring. Router-reward trends and offload rate over training steps.}
  \label{fig:ppo-training}
\end{figure}

\begin{figure}[t]
  \centering
  \includegraphics[width=0.9\linewidth]{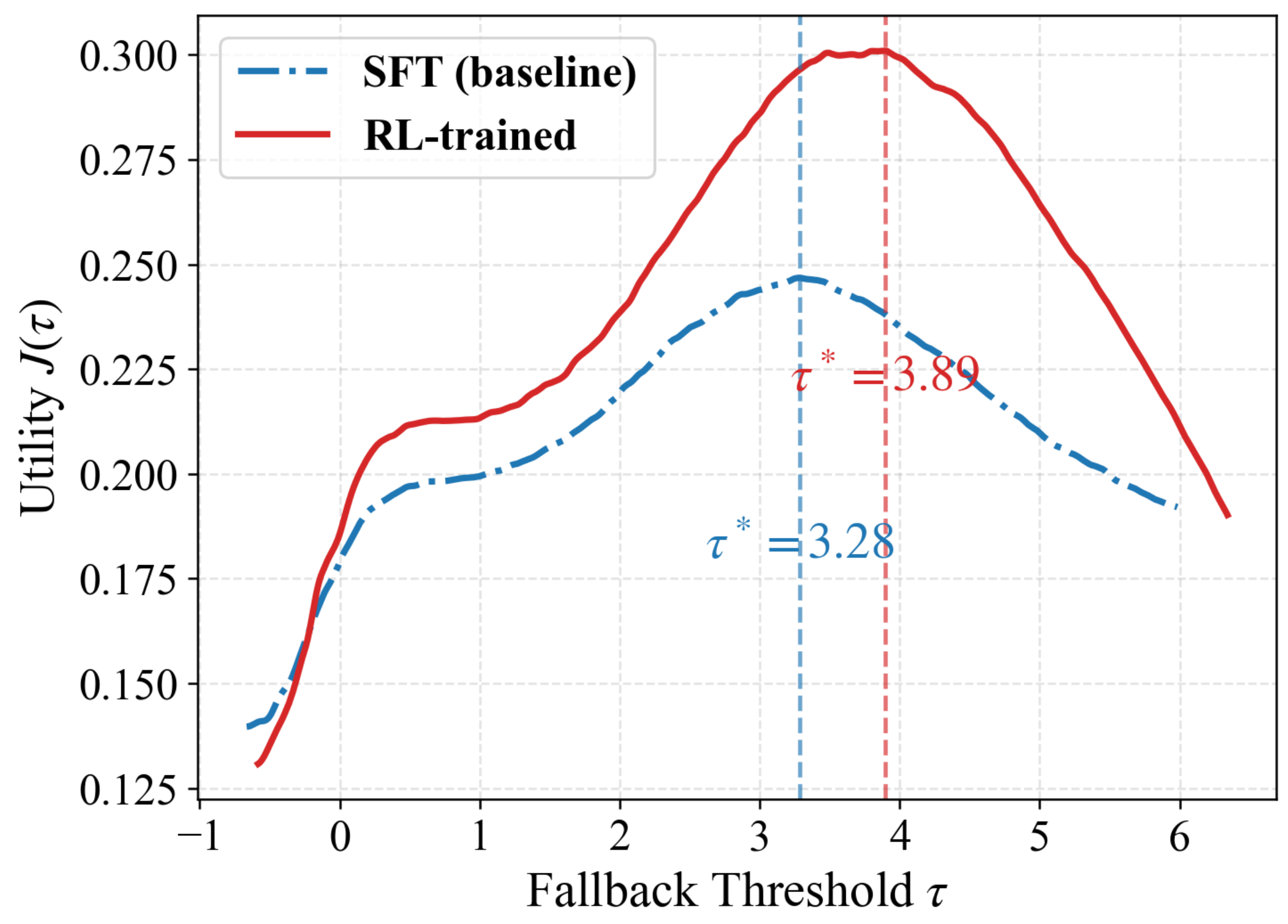}
  \caption{Utility comparison before and after RL.}
  \label{fig:ppo-frontier}
\end{figure}
\begin{figure}[t]
  \centering
  \includegraphics[width=0.9\linewidth]{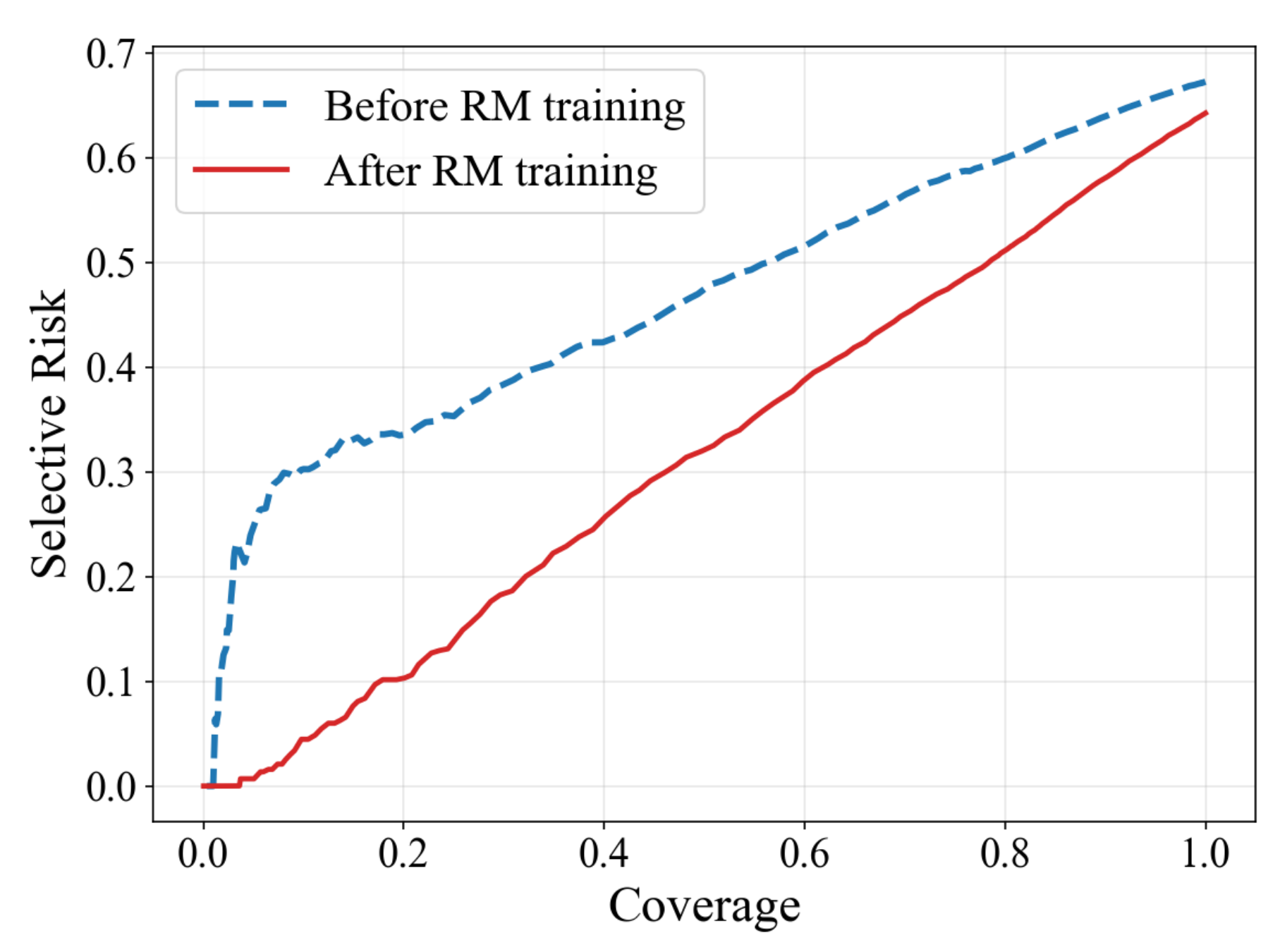}
  \caption{RM Risk-Coverage curves before vs.\ after training. Post-training curves reduce risk at fixed coverage, indicating more accurate router scoring near the decision boundary.}
  \label{fig:RM_improve}
\end{figure}

To quantify the router's scoring quality, we examine risk-coverage curves \cite{el-yaniv2010selective} before and after RM training. Notably, for a fallback threshold $\tau$, the \emph{coverage} is the fraction of requests accepted at the edge, while \emph{selective risk} is the fraction of these accepted cases for which the cloud would have produced a better result (i.e., should-have-offloaded errors). As shown in Fig.~\ref{fig:RM_improve}, the post-training curve lies uniformly below the pre-training curve, especially in the medium-to-high coverage range, demonstrating sharper discrimination near the decision boundary. 

\section{Conclusion and Future Work}
\label{sec:conclusion}
We have introduced a network-aware LLM routing framework that unifies schema-constrained tool planning at the edge with a router-as-reward loop for continual improvement. We have cast offloading as state-dependent thresholding of a unified RM score, aligning routing, cost, and quality through a single, interpretable fallback threshold adapted to the network state. We have implemented two deployable variants—FuncDyn and a lightweight PolicyNet—both adapting online to time-varying links with minimal overhead. We have established the existence and uniqueness of the optimal fallback threshold and derived monotone comparative statics with respect to bandwidth and RTT. Furthermore, we have validated the effectiveness and superiority across extensive simulations.
Subsequent research will investigate multi-edge/multi-cloud orchestration with cooperative routers and a shared reward model, robustness of reward signals under distribution shift, and hardware-in-the-loop evaluations on real network traces.

\begin{appendices}
\begin{lstlisting}[float=t,language=json,caption={SFT sample with structured tool call.}]
{
  "query": "Optimize a federated recommendation system for news articles across mobile devices.",
  "completed_steps": [
    "initialize_nodes"
  ],
  "thought": "After initializing nodes, we need to simulate the network topology. Since this involves mobile devices, we should account for variable bandwidth and latency. This will help model realistic communication conditions during federated training.",
  "output": {
    "name": "simulate_topology",
    "args": {
      "node_count": 5,
      "bandwidth_limits": "10-100 Mbps",
      "latencies": "50-200ms"
    }
  },
  "available_tools": [
    "initialize_nodes(node_count, model_class, dataset_name)",
    "simulate_topology(node_count, bandwidth_limits, latencies)",
    "train_local_models(nodes, dataset, epochs, batch_size, learning_rate)",
    "train_local_models_fedprox(nodes, dataset, global_model, mu, epochs, batch_size, learning_rate)",
    "train_local_models_secure_agg(nodes, dataset, epochs, batch_size, learning_rate)",
    "aggregate_global_model(local_models, weights, nodes, dataset, topology)",
    "aggregate_global_model_krum(local_models, num_selected, nodes, dataset, topology)",
    "aggregate_global_model_median(local_models, nodes, dataset, topology)",
    "aggregate_global_model_secure_agg(local_models, nodes, dataset, topology)",
    "finish()"
  ]
}
\end{lstlisting}

\begin{lstlisting}[float=!ht,language=json,caption={Preference pair: chosen vs. rejected sample.}]
{
  "input": {
    "query": "Stanford medical researchers require federated training of a 3D U-Net for brain tumor segmentation across 7 hospitals, with differential privacy epsilon=0.5 guarantees.",
    "completed_steps": []
  },
  "chosen": {
    "action": {
      "name": "initialize_nodes",
      "args": {
        "node_count": 8,
        "model_class": "3D U-Net",
        "dataset_name": "brain_tumor_segmentation"
      }
    },
    "thought": "First initialize the federated learning nodes with the 3D U-Net model and brain tumor segmentation dataset. Next steps will involve setting up the network topology and beginning the federated training process with differential privacy guarantees. Need to ensure the model class and dataset name are correctly specified for the medical imaging task. Initializing 8 nodes (7 hospitals + central server) with 3D U-Net models for brain tumor segmentation."
  },
  "rejected": {
    "action": {
      "name": "initialize_nodes",
      "args": {
        "node_count": 7,
        "model_class": "MedicalImagingModel",
        "dataset_name": "TumorScans"
      }
    },
    "thought": "Starting fresh with this sensitive medical task. We'll initialize 7 nodes (one per hospital) with a specialized medical imaging model. Using a small batch size (16) to handle potentially limited patient data at each site."
  }
}
\end{lstlisting}
\section{Initialization of Edge Policy, RM, and Fallback Threshold}
\label{sec:initialization}
We initialize the edge LLM policy $\pi_\theta$, the reward/router model $g_\psi$ that outputs the score $s_{i,k}$, and a fixed fallback threshold $\tau_0$ for runtime use. 

Particularly, we fit $\pi_\theta$ on traces $(x_i, H_{i,k}, u_{i,k})$ with supervised cross-entropy:
\begin{align}
\mathcal{L}_{\mathrm{SFT}}(\theta)
&=\mathbb{E}\!\left[-\log \pi_\theta(u_{i,k}\mid x_i, H_{i,k})\right].
\label{eq:SFT}
\end{align}
Here $x_i$ is the task input, $H_{i,k}$ is the historical step, and $u_{i,k}=(\texttt{tool},\texttt{args},\texttt{thought})$ is the teacher action that satisfies the schema. 

The reward model outputs a scalar score $
s_{i,k}=g_\psi(x_i, H_{i,k}, u_{i,k})\in\mathbb{R}$ for routing and as the learning reward. We initialize $g_\psi$ with a pairwise ranking loss:
\begin{align}
\mathcal{L}_{\mathrm{RM}}(\psi)
&=\mathbb{E}\!\left[\log\!\big(1+\exp(-(s^{+}_{i,k}-s^{-}_{i,k}))\big)\right]. 
\label{eq:RM_train}
\end{align}
Here $s^{+}_{i,k}$ and $s^{-}_{i,k}$ are the RM scores of a preferred and a non-preferred answer under the same context $(x_i, H_{i,k})$. The preferred answers come from the cloud LLM outputs, while the non-preferred ones are generated by the edge LLM after SFT initialization. This construction lets the RM learn the quality gap between edge and cloud responses under identical conditions.

We select an initial fallback threshold $\tau_0$ on an offline dataset $\mathcal{D}_{\mathrm{init}}$. Each record in $\mathcal{D}_{\mathrm{init}}$ provides the score $s_{i,k}$ and the empirical utilities we would obtain by choosing edge or cloud at that step:
\begin{align}
J^{\mathrm{edge}}_{i,k} &\triangleq Q^{\mathrm{edge}}_{i,k}-\lambda C^{\mathrm{edge}}_{i,k}, \qquad
J^{\mathrm{cloud}}_{i,k} \triangleq Q^{\mathrm{cloud}}_{i,k}-\lambda C^{\mathrm{cloud}}_{i,k}.
\end{align}
Here $Q^{\cdot}_{i,k}\in[0,1]$ and $C^{\cdot}_{i,k}\ge0$ denote the measured quality and the step-level cost (e.g., latency or monetary cost), respectively. 
$\lambda>0$ is the trade-off weight. We then choose $\tau_0$ by empirical utility maximization under the hard fallback threshold rule:
\begin{align}
\tau_0 \label{eq:initial_tau}\\
\in&  \arg\max_{\tau\in\mathbb{R}}
\sum_{(i,k)\in \mathcal{D}_{\mathrm{init}}}
\Big[J^{\mathrm{edge}}_{i,k}\,\mathbb{I}[s_{i,k}\ge \tau] \notag + J^{\mathrm{cloud}}_{i,k}\,\mathbb{I}[s_{i,k}< \tau]\Big],
\end{align}
which picks the fixed fallback threshold that yields the highest average utility on $\mathcal{D}_{\mathrm{init}}$.

Our training datasets are batch-generated with \texttt{GPT-4o} \cite{openai2024gpt4o}, yielding $8,000$ SFT instances for schema-faithful initialization and $2,000$ preference pairs for reward modeling. In addition, we generate $8,000$ tasks for inference experiments and simulated continual improvement. To ensure diversity, each batch draws a prompt \emph{uniformly at random} from a maintained pool containing multiple prompt templates. Every template rigidly specifies the output schema (tool name, arguments, and brief thought trace) to match our tool-calling interface. After generation, we run a two-stage filtering pass: (i) schema validation (JSON/tool-slot conformance and argument type checks), and (ii) quality screening (automatic scoring to eliminate near-duplicates and low-utility samples). The resulting corpora include: (a) \emph{SFT data}, used to initialize schema-faithful tool use for the edge agent; and (b) \emph{preference pairs}, where a \emph{chosen} action (typically the cloud-produced response) is contrasted with a \emph{rejected} action (typically an edge-produced alternative) under the same context, supporting reward modeling and preference optimization.

\section{Caching Setup and On-Policy Logging}
\label{sec:log}
With \texttt{DeepSeek-V3.2-Exp} as the fallback cloud LLM, the cached records serve multiple purposes in the learning pipeline:
\begin{itemize}
  \item Cloud-offloaded traces provide training data for improving the edge LLM through RL.
  \item Near-threshold cloud cases, together with the small set of above-threshold samples intentionally uploaded, constitute a focused dataset for refining the reward model around the decision boundary.
  \item Execution outputs contribute to prompt construction: edge-side prompts append only the executed tool name to \texttt{completed\_steps}, whereas cloud-side prompts additionally include a concise natural-language summary.
\end{itemize}

Therefore, we denote the cached tuple by
\begin{align}
\mathcal{T}_{i,k} \;\triangleq\; (x_i, H_{i,k}, u_{i,k}, u^{\star}_{i,k}, s_{i,k}). \notag
\end{align}
Only steps offloaded to the cloud (\,$d_{i,k}=\textsc{Cloud}$\,) append $\mathcal{T}_{i,k}$ to the local cache $\mathcal{B}_{\mathrm{RM}}$ and $\mathcal{B}_{\mathrm{RL}}$ for later updates; on-device steps retain only $H_{i,k}$.

To enrich the reward model's learning signal, we randomly sample a small fraction of steps whose scores are slightly above the fallback threshold, i.e., $s_{i,k} \gtrsim \tau(S_{i,k})$, and also upload them to the cloud. These near-threshold samples capture marginal decision cases that are most informative for distinguishing good and bad routing choices.
After executing the selected tool, we record a structured output
\begin{align}
o_{i,k} = (t_{i,k},\,\sigma_{i,k}),
\end{align}
where $t_{i,k}$ is the tool identifier and $\sigma_{i,k}$ is a length-bounded summary of the tool output (sufficient for the next-step LLM to interpret the result while keeping context compact).

The incorporation of $o_{i,k}$ into the local state depends on which model the next prompt is constructed for: edge-side prompts use only the tool identifier, whereas cloud-side prompts include both the identifier and its summary. Formally,
\begin{align}
H_{i,k+1} =
\begin{cases}
H_{i,k} \oplus t_{i,k}, & \text{to edge LLM},\\[3pt]
H_{i,k} \oplus (t_{i,k},\sigma_{i,k}), & \text{to cloud LLM}.
\end{cases}
\label{eq:update}
\end{align}
Here $\oplus$ denotes concatenation into the reasoning trace.

\end{appendices}


\end{document}